\documentclass[aps,prx,reprint,superscriptaddress,longbibliography,
  amsmath,amssymb,floatfix]{revtex4-2}
\usepackage{amsmath,amssymb,amsthm,mathtools}
\usepackage{graphicx,booktabs,array}
\usepackage{tikz}
\usetikzlibrary{arrows.meta,positioning,calc,fit,backgrounds}
\usepackage{microtype}
\usepackage{enumitem}
\usepackage[colorlinks=true,allcolors=blue!45!black]{hyperref}
\usepackage[nameinlink,noabbrev]{cleveref}
\newtheorem{theorem}{Theorem}
\newtheorem{lemma}{Lemma}
\newtheorem{corollary}{Corollary}
\crefname{theorem}{Theorem}{Theorems}
\crefname{lemma}{Lemma}{Lemmas}
\crefname{corollary}{Corollary}{Corollaries}
\crefname{appendix}{Appendix}{Appendixes}

\newcommand{\bits}{\{0,1\}}
\newcommand{\CP}{\operatorname{CP}}
\newcommand{\CSWAP}{\operatorname{CSWAP}}
\newcommand{\ADD}{\operatorname{ADD}}
\newcommand{\MADD}{\operatorname{MADD}}
\newcommand{\Route}{\mathcal{R}}
\newcommand{\Query}{\mathcal{Q}}

\newcommand{\CHW}{C_{\mathrm{HW}}}
\definecolor{inkblue}{RGB}{0,94,135}
\definecolor{lightblue}{RGB}{230,241,247}
\definecolor{amber}{RGB}{165,96,0}
\definecolor{lightamber}{RGB}{255,243,221}
\tikzset{qwire/.style={line width=.55pt},
  qbox/.style={draw=inkblue,fill=lightblue,minimum height=6mm,
    inner sep=3pt,font=\small,align=center},
  dbox/.style={draw=amber,fill=lightamber,minimum height=6mm,
    inner sep=3pt,font=\small,align=center},
  flabel/.style={font=\small,align=left}}
\begin{document}
\title{Exact quantum Hamming-weight computation at optimal depth and minimum width}

\author{Wei Zi}
\email{ziwei.quantum@outlook.com}
\affiliation{Quantum Science Center of Guangdong-Hong Kong-Macao Greater Bay Area, Shenzhen, China}

\author{Pei Yuan}
\email{peiyuan0104@gmail.com}
\affiliation{Independent Researcher, Shenzhen, Guangdong, China}

\author{Junhong Nie}
\email[Corresponding author: ]{123miyoi@gmail.com}
\affiliation{Shandong University, Jinan, Shandong, China}

\begin{abstract}
Hamming-weight computation and symmetric Boolean functions are
basic primitives for quantum arithmetic and state preparation,
where circuit depth and qubit overhead matter.
We construct an exact circuit computing the Hamming weight of
an $n$-qubit register with depth $O(\log n)$ and size
$O(n\log n)$, using only the $n$ input qubits and
$m=\lceil\log_2(n+1)\rceil$ binary output qubits.
The circuit uses arbitrary one-qubit gates and CNOTs with
all-to-all connectivity.
Both depth and total width are optimal.
Our constructions combine tree structures and pipelining,
temporarily using input qubits as workspace.
We also give exact circuits for every symmetric Boolean function
with the same depth and size bounds, using only $m$ ancillary
qubits returned to zero.
Finally, we apply the counter to exact preparation of balanced
Dicke states for even $n$.
With only $O(\log n)$ ancillary qubits, we obtain expected depth
$O(\log^2 n)$ using measurements and classical feedforward,
and depth $O(n^{1/4}\log n)$ using unitary gates alone.
\end{abstract}

\maketitle

\section{Introduction}
\label{sec:introduction}

Many quantum algorithms need to determine the Hamming weight
of a bit string without resolving its individual bits. A coherent
counter records this weight in a binary register while preserving
the input and acting coherently on superpositions. Such counters appear in
entanglement concentration~\cite{KayeMosca2001}, in arithmetic
constructions that reduce the cost of applying equal phase
rotations~\cite{Gidney2018}, and in spin-sector projection for
quantum state preparation~\cite{KoboriEtAl2026}.
Although an $n$-bit input has only $n+1$ possible weights,
computing the weight requires information from every input bit.
This makes counting a basic setting for studying how much
workspace is needed to achieve small circuit depth.
Related depth--space questions arise in state preparation and
unitary synthesis~\cite{SunEtAl2023,YuanZhang2023,ZhangLiYuan2022}.

Several constructions address the depth--space tradeoff for
unitary counting. Standard reversible carry-save counting
achieves depth $O(\log n)$ with $O(n)$ clean ancillary qubits
\cite{Gossett1998,ZiYuanNieZhang2026}.
Takahashi and Tani~\cite{TakahashiTani2021} obtained depth
$O(\log^2 n)$ with $O(\log n)$ initialized qubits and
$O(n\log^2 n)$ uninitialized qubits.
Zi, Nie, and Sun~\cite{ZiNieSun2025} removed this additional
workspace while retaining depth $O(\log^2 n)$.
More recently, Zi et al.~\cite{ZiYuanNieZhang2026} obtained
depth $O(\log n)$ with sublinear additional clean workspace.
Clean qubits are initialized to zero; uninitialized, or dirty,
qubits may start in arbitrary states and must be restored after use.

Measurement-assisted constructions provide another route to
shallow counting. Piroli et al.~\cite{PiroliStyliarisCirac2024}
and Rethinasamy et al.~\cite{RethinasamyLaBordeWilde2024}
gave logarithmic-depth Hamming-weight measurements with linear
ancillary space, preserving coherence within each weight
subspace. These constructions use intermediate measurements
and classical feedforward.
For circuits built from elementary one- and two-qubit gates
without measurements, the central question remains:
can optimal depth and minimum width be attained simultaneously,
using only the input and the prescribed binary output register?

We answer this question affirmatively. Let $\mathsf{X}$ contain the $n$
input qubits, and let $\mathsf{W}$ contain
$m=\lceil\log_2(n+1)\rceil$ output qubits. Writing $|x|$ for
the Hamming weight of $x\in\{0,1\}^n$, our circuit implements
\begin{equation}
  \CHW:\quad |x\rangle_\mathsf{X}|0^m\rangle_\mathsf{W}
  \longmapsto |x\rangle_\mathsf{X}\bigl||x|\bigr\rangle_\mathsf{W}
  \label{eq:intro-hamming-map}
\end{equation}
exactly, with depth $O(\log n)$, size $O(n\log n)$, and total
width $n+m$. It uses arbitrary one-qubit gates and CNOTs with
all-to-all connectivity and no measurements. The output must
distinguish $n+1$ weights, while its least significant bit computes
parity and therefore requires depth $\Omega(\log n)$.
Thus both depth and width are optimal for this
binary-output interface. Table~\ref{tab:counting} summarizes
the relevant depth--workspace bounds.

\begin{table*}[t]
\centering
\caption{Selected constructions for exact binary Hamming-weight
computation with elementary one- and two-qubit gates, all-to-all
connectivity, and no measurements.
Additional workspace excludes the $n$ input qubits and the
$m=\lceil\log_2(n+1)\rceil$ prescribed output qubits.
All workspace is restored.
The carry-save row uses the standard reversible population-count
construction~\cite{Gossett1998}; see also
Ref.~\cite[Lemma~2.5]{ZiYuanNieZhang2026}.
The Takahashi--Tani row includes compute--copy--uncompute
to obtain the fully coherent interface, with $O(\log n)$
additional clean qubits.
For Zi et al.~\cite{ZiYuanNieZhang2026}, we choose
$B=2^{\sqrt{\log_2 n}}$ in their tradeoff, up to integer rounding.
The present circuit has size $O(n\log n)$.}
\label{tab:counting}
\begingroup
\small
\setlength{\tabcolsep}{7pt}
\renewcommand{\arraystretch}{1.2}
\begin{tabular*}{\textwidth}{@{\extracolsep{\fill}}lccc@{}}
\toprule
Construction & Clean workspace & Dirty workspace & Depth \\
\midrule
Reversible carry-save counting~\cite{Gossett1998}
  & $O(n)$ & $0$ & $O(\log n)$ \\
Takahashi and Tani~\cite{TakahashiTani2021}
  & $O(\log n)$ & $O(n\log^2 n)$ & $O(\log^2 n)$ \\
Zi, Nie, and Sun~\cite{ZiNieSun2025}
  & $0$ & $0$ & $O(\log^2 n)$ \\
Zi et al.~\cite{ZiYuanNieZhang2026}
  & $O\!\left(n\sqrt{\log_2 n}/2^{\sqrt{\log_2 n}}\right)$
  & $0$ & $O(\log n)$ \\
This work (Theorem~\ref{thm:hamming-weight})
  & $0$ & $0$ & $O(\log n)$ \\
\bottomrule
\end{tabular*}
\endgroup
\end{table*}

Our second result concerns symmetric Boolean functions, whose
values depend only on the input weight. These functions have
served as a setting for studying computation with limited
workspace~\cite{MaslovEtAl2021,TakahashiTani2021} and the
power of shallow quantum circuits~\cite{XuLi2026Fanout}.
For every symmetric Boolean function
$f:\{0,1\}^n\to\{0,1\}$, we construct an exact circuit implementing
\[
  |x\rangle|t\rangle
  \longmapsto |x\rangle|t\oplus f(x)\rangle,
  \qquad t\in\{0,1\},
\]
with depth $O(\log n)$ and size $O(n\log n)$, using only
$m=\lceil\log_2(n+1)\rceil$ ancillary qubits, all initialized
and returned to $|0\rangle$.
The depth is optimal in the worst case, since symmetric
Boolean functions include parity.
Our counting and Boolean-evaluation circuits both combine
tree structures with pipelining across tree levels.

We apply these circuits to the preparation of Dicke states,
the uniform superpositions of bit strings with a fixed weight.
In the measurement-assisted setting, measuring the counter's output
implements a Hamming-weight measurement using only $m$
ancillary qubits. Substituting it into the adaptive strategy of
Yu et al.~\cite{YuEtAl2026} prepares the balanced Dicke state
exactly for even $n$, with expected depth $O(\log^2 n)$.
Compared with implementing the same strategy using the earlier
counter~\cite{ZiNieSun2025}, this reduces the expected depth
from $O(\log^3 n)$ without increasing the ancillary space.
Here single-qubit measurements and resets have constant depth,
and classical feedforward latency is excluded.
With linear ancillary space,
Yang et al.~\cite{YangWeiYin2026SpinSqueezing} obtain
$O(\log n)$ expected depth under repetition; our result retains
logarithmic ancillary space.

Weight-selective reflections also connect the counter to
amplitude-amplification approaches to Dicke-state
preparation~\cite{PiroliStyliarisCirac2024,NagibSaffmanMolmer2025}.
In the unitary model, this gives exact preparation at depth
$O(r^{1/4}\log n)$,
where $r=\min\{k,n-k\}$ and $1\le k\le n-1$ is a
specified classical weight, using $m$ ancillary qubits that
are returned to zero. For balanced states, the resulting depth
$O(n^{1/4}\log n)$ improves the $O(n)$ bound of the
ancilla-free construction of Yuan and
Zhang~\cite{YuanZhang2025Dicke}, while total width increases
only from $n$ to $n+m=(1+o(1))n$.
Other constructions achieve polylogarithmic depth with larger
ancillary spaces~\cite{LiuChildsGottesman2026}.
The detailed comparisons are
given in Sec.~\ref{sec:dicke-states}.

Section~\ref{sec:preliminaries} specifies the circuit model.
Sections~\ref{sec:hamming-weight} and \ref{sec:symmetric-functions}
develop the Hamming-weight counter and circuits for symmetric Boolean functions, followed by the
Dicke-state applications in Sec.~\ref{sec:dicke-states} and
discussion in Sec.~\ref{sec:discussion}. The appendixes give
the detailed arithmetic and Boolean-evaluation schedules.

\section{Circuit model and notation}
\label{sec:preliminaries}

\paragraph{Circuit model and resources.}
Unless stated otherwise, we use exact unitary circuits with
all-to-all connectivity, built from arbitrary one-qubit gates
and CNOTs. Gates in a depth layer have pairwise disjoint
supports. For a circuit $U$, let $D(U)$ denote its depth and
$S(U)$ its number of elementary unitary gates. Width counts
all physical qubits; additional workspace excludes the
specified input and output registers. In our pipeline schedules,
a logical round consists of a constant number of depth layers.
Controlled phases, SWAPs, controlled-SWAPs, and controlled
$X$ gates with a fixed number of controls have exact
constant-size, constant-depth decompositions without ancillary
qubits~\cite{BarencoEtAl1995}. We use these gates as shorthand
and write $\operatorname{CCX}_{\mathsf{A},\mathsf{B}\to \mathsf{T}}$ for a Toffoli gate
with controls $\mathsf{A},\mathsf{B}$ and target $\mathsf{T}$.
The adaptive application in Sec.~\ref{subsec:dicke-adaptive}
also allows measurements and resets, with the depth convention
specified there; $S$ continues to count elementary unitary gates.

\paragraph{Workspace and exact implementation.}
A clean qubit starts in $|0\rangle$, and reusable clean workspace
is returned to zero after use. A prescribed output register need
not be restored: the weight register $\mathsf{W}$ is an output in counting
and temporary clean workspace in Boolean evaluation.
Dirty qubits may start in arbitrary states, including entanglement
with other systems. Borrowed workspace is taken from registers
already present and therefore does not increase the width.
At the end of each borrowing period, the net operation acts as
the identity on the borrowed subsystem, including when it is
initially entangled with other registers.
All stated computational-basis maps specify the full coherent
action, with no additional input-dependent phases; they extend
linearly to arbitrary superpositions and entangled inputs.

\paragraph{Binary notation.}
All logarithms are base two. We write $[n]=\{1,\ldots,n\}$
and $|x|=\sum_{i=1}^n x_i$ for the Hamming weight of
$x\in\{0,1\}^n$. An $m$-qubit word uses little-endian
bit positions $j=0,\ldots,m-1$, with integer label
\[
  z=\sum_{j=0}^{m-1}2^jz_j,
  \qquad N=2^m,
  \qquad \omega=e^{2\pi i/N}.
\]
We identify the bit string $(z_0,\ldots,z_{m-1})$ with
the integer $z$ when evaluating Boolean functions.
Word additions are modulo $N$, and $\oplus$ denotes addition
modulo two. Operator products act from right to left; explicit
gate sequences and circuit diagrams are read in execution order
from left to right. For a real angle $\theta$, we use
\[
  \CP(\theta)=\operatorname{diag}(1,1,1,e^{i\theta}).
\]

\paragraph{Fourier transforms.}
For binary weight readout, we use the Fourier transform
\[
  F_m|z\rangle
    =\frac{1}{\sqrt N}\sum_{q=0}^{N-1}\omega^{qz}|q\rangle.
\]
For the addition pipeline, it is convenient to use the
bit-reversed form, which displays the phase on each qubit:
\begin{align}
  |\phi_j(z)\rangle
    &:=\frac{|0\rangle+e^{2\pi iz/2^{j+1}}|1\rangle}{\sqrt2},
    \label{eq:fourier-qubit}\\
  \widetilde F_m|z\rangle
    &:=\bigotimes_{j=0}^{m-1}|\phi_j(z)\rangle,
    \label{eq:product-fourier}
\end{align}
where factor $j$ is assigned to qubit $j$. If
$\operatorname{Rev}_m$ reverses the order of the $m$ qubits,
then
\begin{equation}
  \widetilde F_m=\operatorname{Rev}_mF_m,
  \qquad F_m^{-1}=\widetilde F_m^{-1}\operatorname{Rev}_m.
  \label{eq:fourier-reversal}
\end{equation}
The reversal uses disjoint SWAPs, with depth $O(1)$ and size
$O(m)$. Both Fourier transforms and their inverses have exact
circuits of depth $O(m)$ and size $O(m^2)$ without ancillary
qubits~\cite{Draper2000}. Section~\ref{subsec:tree-mechanism}
uses the individual qubit phases to decode sums progressively;
Appendix~\ref{sec:simultaneous-addition} gives the gate schedule.

\section{Space-optimal counting in logarithmic depth}
\label{sec:hamming-weight}

Let \(\mathsf{X}=(\mathsf{X}_1,\ldots,\mathsf{X}_n)\) be the data register and put
\[
  m=\left\lceil\log_2(n+1)\right\rceil,
  \qquad N=2^m.
\]
The prescribed \(m\)-qubit output register is denoted by \(\mathsf{W}\).

\begin{theorem}[Exact Hamming-weight computation]
  \label{thm:hamming-weight}
  There is an exact circuit implementing
  \[
    |x\rangle_\mathsf{X}|0^m\rangle_\mathsf{W}
      \longmapsto
    |x\rangle_\mathsf{X}\bigl||x|\bigr\rangle_\mathsf{W}
  \]
  with depth \(O(\log n)\), size \(O(n\log n)\), and no qubits beyond
  \(\mathsf{X},\mathsf{W}\). Every circuit of one- and two-qubit gates implementing this
  map has depth at least \(\lceil\log_2 n\rceil\).
\end{theorem}

The core ingredient is the following simultaneous modular addition
lemma.

\begin{lemma}[Simultaneous modular addition]
  \label{thm:simultaneous-addition}
  Let \(m,r\ge1\), and let \(\mathsf{Q},\mathsf{B}_1,\ldots,\mathsf{B}_r\) be pairwise disjoint
  \(m\)-qubit registers. There is an exact circuit
  \begin{equation}
    \MADD_{m,r}:\quad
    |q\rangle_\mathsf{Q}\bigotimes_{v=1}^r|b_v\rangle_{\mathsf{B}_v}
      \longmapsto
    |q\rangle_\mathsf{Q}\bigotimes_{v=1}^r|b_v+q\rangle_{\mathsf{B}_v},
    \label{eq:madd-map}
  \end{equation}
  where each target addition is modulo \(2^m\), with depth
  \(O(m+\log r)\), size \(O(rm^2)\), and no qubits beyond
  \(\mathsf{Q},\mathsf{B}_1,\ldots,\mathsf{B}_r\).
\end{lemma}

The role of simultaneous modular addition in Hamming-weight
computation is not immediately apparent from its statement.
Computing the Hamming weight requires the small output register
\(\mathsf{W}\) to collect information from all \(n\) input qubits.
Achieving logarithmic depth without additional qubits therefore
poses a challenge for parallelizing this information collection.

Lemma~\ref{thm:simultaneous-addition} provides the key primitive:
in this application its source $\mathsf{Q}$ is the output register $\mathsf{W}$.
It allows us to add that register simultaneously to many
dirty word registers borrowed from the input.
Our construction uses these borrowed registers to collect
contributions from different input blocks in parallel and combine
them in the output, while restoring all borrowed qubits.
Thus the input itself supplies the temporary registers needed
for parallel information collection.

We first show how Lemma~\ref{thm:simultaneous-addition} implies
Theorem~\ref{thm:hamming-weight}, and then construct the simultaneous
modular adder. Its complete gate schedule is given in
Appendix~\ref{sec:simultaneous-addition}.
Figure~\ref{fig:counting} summarizes the counting circuit and its
restoration mechanism.

\subsection{Hamming-weight computation from simultaneous addition}
\label{subsec:counting-from-addition}

\begin{figure*}[t]
\centering
\resizebox{\textwidth}{!}{%
\begin{tikzpicture}[x=1cm,y=1cm]
\node[flabel,anchor=west] at (0,.45)
 {\textbf{(a)} Exact counting on $n+m$ qubits};

\draw[qwire] (1.25,0)--(15.4,0);
\draw[qwire] (1.25,-1)--(15.4,-1);

% Register-width markers.
\foreach \yy/\width in {0/n,-1/m}{
 \draw[line width=.6pt]
   (1.40,\yy-.10)--(1.54,\yy+.10);
 \node[anchor=south,font=\scriptsize,inner sep=1pt]
   at (1.47,\yy+.11) {$\width$};
}

\node[anchor=east,font=\small] at (1.15,0)
 {$\mathsf{X}:\ |x\rangle$};
\node[anchor=east,font=\small] at (1.15,-1)
 {$\mathsf{W}:\ |0^m\rangle$};

\node[qbox] at (2.45,-1) {$H^{\otimes m}$};
\foreach \xx/\aa in {4.2/0,7.0/1,9.8/2}{
 \node[qbox,minimum width=1.7cm,minimum height=1.55cm]
   at (\xx,-.5)
   {$\mathcal E_{\aa}$\\[-1pt]\scriptsize collect batch $I_{\aa}$};
}
\node[qbox] at (12.6,-1) {$F_m^{-1}$};

\node[anchor=west,font=\small] at (15.45,0)
 {$|x\rangle$};
\node[anchor=west,font=\small] at (15.45,-1)
 {$\bigl||x|\bigr\rangle$};

\node[flabel,anchor=west] at (0,-2.3)
 {\textbf{(b)} One block: translate, phase, restore, cancel};

\foreach \yy/\lab in {0/{\mathsf{C}:\ x_\mathsf{C}},-.8/{\mathsf{B}:\ b},-1.6/{\mathsf{W}:\ q}}{
 \draw[qwire] (1.5,\yy-3)--(14.8,\yy-3);
 \draw[line width=.6pt]
   (1.88,\yy-3.10)--(2.02,\yy-2.90);
 \node[anchor=south,font=\scriptsize,inner sep=1pt]
   at (1.95,\yy-2.89) {$m$};
 \node[anchor=east,font=\small] at (1.4,\yy-3)
   {$\lab$};
}

\node[qbox,minimum width=1.5cm,minimum height=1.22cm]
 at (3.3,-4.2) {$\ADD_{\mathsf{W}\to \mathsf{B}}$};
\node[dbox,minimum width=1.5cm,minimum height=1.22cm]
 at (6,-3.4) {$\mathcal D_{\mathsf{C},\mathsf{B}}$};
\node[qbox,minimum width=1.5cm,minimum height=1.22cm]
 at (9,-4.2) {$\ADD_{\mathsf{W}\to \mathsf{B}}^{-1}$};
\node[dbox,minimum width=1.5cm,minimum height=1.22cm]
 at (12,-3.4) {$\mathcal D_{\mathsf{C},\mathsf{B}}^{-1}$};

\node[font=\small] at (4.5,-5.15)
 {$b\mapsto b+q$};
\node[font=\small] at (7.5,-5.15)
 {$\omega^{(b+q)|x_\mathsf{C}|}$};
\node[font=\small] at (10.5,-5.15)
 {$b+q\mapsto b$};
\node[font=\small] at (13.5,-5.15)
 {$\omega^{-b|x_\mathsf{C}|}$};
\node[anchor=west,font=\small,align=left] at (15,-3.8)
 {$\mathsf{B}$ restored\\phase $\omega^{q|x_\mathsf{C}|}$};
\end{tikzpicture}%
}
\caption{Coherent counting using the input as temporary workspace.
(a) A Hadamard layer prepares the output, three batches collect the
weight-dependent phase, and an inverse Fourier transform reads out
the binary weight. Each batch borrows registers from its complement
and restores them before the next batch.
(b) The input block contains the string $x_\mathsf{C}$, the dirty word has
integer label $b$, and the output has integer label $q$.
Here $\ADD_{\mathsf{W}\to \mathsf{B}}$ adds the output to the dirty word modulo $N$.
The two phase operations cancel the dependence on $b$, leaving
$\omega^{q|x_\mathsf{C}|}$. For a full batch, $\mathcal M_a$ replaces
$\ADD_{\mathsf{W}\to \mathsf{B}}$ and translates all borrowed words simultaneously.}
\label{fig:counting}
\end{figure*}
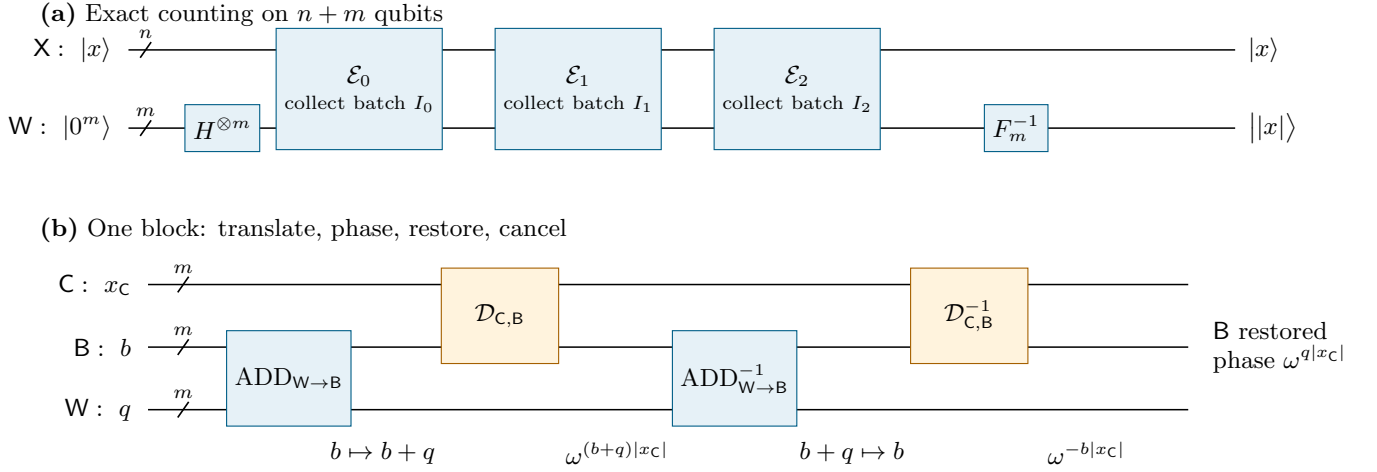

We prove Theorem~\ref{thm:hamming-weight} assuming
Lemma~\ref{thm:simultaneous-addition}.

\paragraph{The Fourier target.}
Starting with \(|0^m\rangle_\mathsf{W}\), apply \(H^{\otimes m}\) to prepare
a uniform superposition. For a data basis state \(|x\rangle_\mathsf{X}\), our
goal is to collect the phase
\begin{equation}
  |x\rangle_\mathsf{X}\frac{1}{\sqrt N}\sum_{q=0}^{N-1}|q\rangle_\mathsf{W}
  \longmapsto
  |x\rangle_\mathsf{X}\frac{1}{\sqrt N}\sum_{q=0}^{N-1}
    \omega^{q|x|}|q\rangle_\mathsf{W}.
  \label{eq:hamming-fourier-target}
\end{equation}
The output state on the right is \(F_m\bigl||x|\bigr\rangle_\mathsf{W}\).
Applying \(F_m^{-1}\) therefore gives the binary weight.
Since \(|x|\le n<N\), no modular ambiguity occurs. It remains to
implement the phase in Eq.~\eqref{eq:hamming-fourier-target} using
only \(\mathsf{X},\mathsf{W}\).

\paragraph{A three-register example.}
Consider three disjoint \(m\)-qubit registers: the output \(\mathsf{W}\), a
dirty word \(\mathsf{B}\), and an input block \(\mathsf{C}\). Write their basis labels
as \(q,b,x_\mathsf{C}\), respectively. Here \(q\) and \(b\) are binary
integers, while \(x_\mathsf{C}\) is the original input string and \(|x_\mathsf{C}|\)
is its Hamming weight. We first construct an operation that leaves
these labels unchanged and contributes the phase
\(\omega^{q|x_\mathsf{C}|}\).

Let \(\ADD_{\mathsf{W}\to \mathsf{B}}\) add \(q\) to \(b\) modulo \(N\)
(the \(r=1\) case of Lemma~\ref{thm:simultaneous-addition}).
To obtain a phase depending on the contents of \(\mathsf{B}\) and \(\mathsf{C}\),
number their qubits by \(0,\ldots,m-1\) and define
\begin{equation}
  \mathcal D_{\mathsf{C},\mathsf{B}}:=
  \prod_{i=0}^{m-1}\prod_{j=0}^{m-1}
  \CP_{\mathsf{C}_i,\mathsf{B}_j}\!\left(\frac{2\pi 2^j}{N}\right).
  \label{eq:block-phase-layer}
\end{equation}
Each occupied input qubit contributes \(\omega^b\), so
\begin{equation}
  \mathcal D_{\mathsf{C},\mathsf{B}}|x_\mathsf{C}\rangle_\mathsf{C}|b\rangle_\mathsf{B}
    =\omega^{b|x_\mathsf{C}|}|x_\mathsf{C}\rangle_\mathsf{C}|b\rangle_\mathsf{B}.
  \label{eq:block-phase-action}
\end{equation}
Thus the input bits contribute directly, without first computing
their weight in a separate register.

Apply
\(\ADD_{\mathsf{W}\to \mathsf{B}};\mathcal D_{\mathsf{C},\mathsf{B}};\ADD_{\mathsf{W}\to \mathsf{B}}^{-1};\mathcal D_{\mathsf{C},\mathsf{B}}^{-1}\),
in execution order from left to right, as shown in
Fig.~\ref{fig:counting}(b). In register order \(\mathsf{C},\mathsf{B},\mathsf{W}\), the
successive states are
\begin{equation}
  \begin{aligned}
  &|x_\mathsf{C},b,q\rangle\\
  \xrightarrow{\ADD_{\mathsf{W}\to \mathsf{B}}}\quad
  &|x_\mathsf{C},[b+q]_N,q\rangle\\
  \xrightarrow{\mathcal D_{\mathsf{C},\mathsf{B}}}\quad
  &\omega^{[b+q]_N|x_\mathsf{C}|}|x_\mathsf{C},[b+q]_N,q\rangle\\
  \xrightarrow{\ADD_{\mathsf{W}\to \mathsf{B}}^{-1}}\quad
  &\omega^{[b+q]_N|x_\mathsf{C}|}|x_\mathsf{C},b,q\rangle\\
  \xrightarrow{\mathcal D_{\mathsf{C},\mathsf{B}}^{-1}}\quad
  &\omega^{q|x_\mathsf{C}|}|x_\mathsf{C},b,q\rangle,
  \end{aligned}
  \label{eq:one-block-cancellation}
\end{equation}
where \([b+q]_N\) denotes the residue modulo \(N\). The last step
uses \(\omega^N=1\): the two phases multiply to
\(\omega^{([b+q]_N-b)|x_\mathsf{C}|}=\omega^{q|x_\mathsf{C}|}\).
The borrowed word returns to \(b\), and its unknown initial contents
no longer appear in the net phase.

The phase circuit contains \(m^2\) controlled-phase gates. In layer
\(t=0,\ldots,m-1\), pair \(\mathsf{C}_i\) with \(\mathsf{B}_{(i+t)\bmod m}\).
Each layer is a matching, and the \(m\) layers cover every pair.
The same construction works for a block of \(s\le m\) input qubits
by retaining only \(i=0,\ldots,s-1\). Consequently,
\begin{equation}
  D(\mathcal D_{\mathsf{C},\mathsf{B}})=O(m),\qquad S(\mathcal D_{\mathsf{C},\mathsf{B}})=O(sm),
  \label{eq:block-phase-resources}
\end{equation}
with no additional qubits. The complete four-step operation has
depth \(O(m)\) and size \(O(m^2)\).

\paragraph{Parallel collection from the full input.}
For sufficiently large \(n\), split the input into blocks of \(m\) qubits,
with a final shorter block if necessary. Distribute the complete
blocks as evenly as possible among three batches, and assign the
short block, if present, to a batch with the fewest complete blocks.
Let \(I_a\) denote the input indices in batch \(a\), so that
\begin{equation}
  [n]=I_0\mathbin{\dot\cup}I_1\mathbin{\dot\cup}I_2.
  \label{eq:three-batches}
\end{equation}
Write the blocks in batch \(a\) as
\(\mathsf{C}_{a,1},\ldots,\mathsf{C}_{a,r_a}\).
The other two batches provide enough complete \(m\)-qubit blocks
to borrow disjoint words \(\mathsf{B}_{a,1},\ldots,\mathsf{B}_{a,r_a}\), one for
each source block. The short block is used only as a source,
never as a borrowed word.

Let \(\mathcal M_a\) be the circuit of
Lemma~\ref{thm:simultaneous-addition} with source \(\mathsf{W}\) and these
\(r_a\) target words, and set
\begin{equation}
  \begin{aligned}
    \mathcal D_a&:=\prod_{g=1}^{r_a}\mathcal D_{\mathsf{C}_{a,g},\mathsf{B}_{a,g}},\\
    \mathcal E_a&:=\mathcal D_a^{-1}\mathcal M_a^{-1}\mathcal D_a\mathcal M_a.
  \end{aligned}
  \label{eq:batch-cancellation}
\end{equation}
The factors in \(\mathcal D_a\) act on disjoint registers and run in parallel.
Applying the preceding cancellation to each block gives
\begin{equation}
  \mathcal E_a|x\rangle_\mathsf{X}|q\rangle_\mathsf{W}
  =\omega^{q\sum_{i\in I_a}x_i}|x\rangle_\mathsf{X}|q\rangle_\mathsf{W}.
  \label{eq:batch-operator}
\end{equation}
Thus one batch contributes its weight to the output phase and
restores all borrowed input qubits. The three batches are processed
sequentially, so their borrowed words may be reused.
Their phases multiply to \(\omega^{q|x|}\), establishing
Eq.~\eqref{eq:hamming-fourier-target}. The complete counting circuit
is therefore
\begin{equation}
  (F_m^{-1})_\mathsf{W} \mathcal E_2 \mathcal E_1 \mathcal E_0 H_\mathsf{W}^{\otimes m},
  \label{eq:hamming-readout}
\end{equation}
as summarized in Fig.~\ref{fig:counting}(a).

\paragraph{Resource bounds.}
For each batch, the simultaneous additions and the parallel phase
circuits give
\[
  D(\mathcal E_a)=O(m+\log r_a),\qquad S(\mathcal E_a)=O(r_am^2).
\]
There are three sequential batches and
\(\sum_a r_a\le n/m+3\). The Hadamard preparation has depth one,
and the inverse Fourier transform has depth \(O(m)\) and size
\(O(m^2)\). Hence
\begin{equation}
  \begin{aligned}
    D&=O\!\left(m+\max_a\log r_a\right)=O(\log n),\\
    S&=O\!\left(m^2\sum_a r_a\right)=O(n\log n).
  \end{aligned}
  \label{eq:hamming-resources}
\end{equation}
Only the \(n\) input qubits and the \(m\) prescribed output qubits
are used. Finally, the least significant output bit is the parity
of all \(n\) inputs, so the standard depth lower bound gives
\(D\ge\lceil\log_2 n\rceil\)~\cite[Lemma~3.1]{FangEtAl2006}.
This completes the proof of Theorem~\ref{thm:hamming-weight}.

\subsection{Simultaneous addition by a Fourier pipeline}
\label{subsec:tree-mechanism}

We now construct the simultaneous adder of
Lemma~\ref{thm:simultaneous-addition}. The construction has two
ingredients. First, two traversals of a binary tree add the source
value to every target while canceling the targets' initial values
from the propagation. Second, a Fourier implementation allows
successive tree levels to work on different bits at the same time.
This reduces the depth from \(O(m\log r)\) for a level-by-level
implementation to \(O(m+\log r)\).

\paragraph{Simultaneous addition by two tree traversals.}

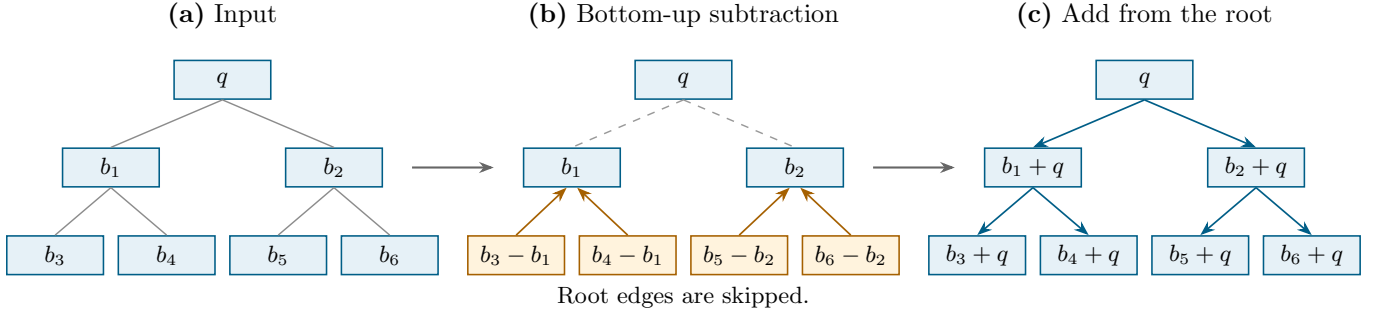
\begin{figure*}[t]
\centering
\resizebox{\textwidth}{!}{%
\begin{tikzpicture}[x=1cm,y=1cm,
  word/.style={draw=inkblue,line width=.6pt,fill=lightblue,
    minimum width=1.30cm,minimum height=.52cm,inner sep=2pt,
    font=\fontsize{9}{11}\selectfont},
  subtractword/.style={word,draw=amber,fill=lightamber},
  treeedge/.style={draw=black!45,line width=.55pt},
  sweep/.style={-{Stealth[length=2mm]},line width=.65pt},
  paneltitle/.style={font=\fontsize{10}{12}\selectfont,anchor=south}]
\foreach \shift/\stage in {0/0,6.2/1,12.4/2}{
  \begin{scope}[xshift=\shift cm]
  \ifnum\stage=0
    \node[paneltitle] at (2.7,.58) {\textbf{(a)} Input};
    \node[word] (q\stage) at (2.7,0) {$q$};
    \node[word] (b1\stage) at (1.2,-1.18) {$b_1$};
    \node[word] (b2\stage) at (4.2,-1.18) {$b_2$};
    \foreach \xx/\vv in {.45/3,1.95/4,3.45/5,4.95/6}
      \node[word] (b\vv\stage) at (\xx,-2.36) {$b_{\vv}$};
    \foreach \vv in {1,2}
      \draw[treeedge] (q\stage.south)--(b\vv\stage.north);
    \foreach \pp/\vv in {1/3,1/4,2/5,2/6}
      \draw[treeedge] (b\pp\stage.south)--(b\vv\stage.north);
  \fi
  \ifnum\stage=1
    \node[paneltitle] at (2.7,.58)
      {\textbf{(b)} Bottom-up subtraction};
    \node[word] (q\stage) at (2.7,0) {$q$};
    \node[word] (b1\stage) at (1.2,-1.18) {$b_1$};
    \node[word] (b2\stage) at (4.2,-1.18) {$b_2$};
    \foreach \xx/\vv/\pp in {
      .45/3/1,1.95/4/1,3.45/5/2,4.95/6/2}
      \node[subtractword] (b\vv\stage) at (\xx,-2.36)
        {$b_{\vv}-b_{\pp}$};
    \foreach \vv in {1,2}
      \draw[treeedge,dashed,black!40]
        (q\stage.south)--(b\vv\stage.north);
    \foreach \pp/\vv/\offset in {
      1/3/-2,1/4/2,2/5/-2,2/6/2}
      \draw[sweep,draw=amber]
        (b\vv\stage.north)--
        ([xshift=\offset pt]b\pp\stage.south);
    \node[font=\fontsize{9}{11}\selectfont] at (2.7,-2.92)
      {Root edges are skipped.};
  \fi
  \ifnum\stage=2
    \node[paneltitle] at (2.7,.58)
      {\textbf{(c)} Add from the root};
    \node[word] (q\stage) at (2.7,0) {$q$};
    \node[word] (b1\stage) at (1.2,-1.18) {$b_1+q$};
    \node[word] (b2\stage) at (4.2,-1.18) {$b_2+q$};
    \foreach \xx/\vv in {.45/3,1.95/4,3.45/5,4.95/6}
      \node[word] (b\vv\stage) at (\xx,-2.36) {$b_{\vv}+q$};
    \foreach \vv in {1,2}
      \draw[sweep,draw=inkblue]
        (q\stage.south)--(b\vv\stage.north);
    \foreach \pp/\vv in {1/3,1/4,2/5,2/6}
      \draw[sweep,draw=inkblue]
        (b\pp\stage.south)--(b\vv\stage.north);
  \fi
  \end{scope}
}
\draw[sweep,draw=black!60,line width=.8pt]
  (5.25,-1.18)--(6.35,-1.18);
\draw[sweep,draw=black!60,line width=.8pt]
  (11.45,-1.18)--(12.55,-1.18);
\end{tikzpicture}%
}
\caption{Tree construction of simultaneous addition, illustrated with
six target words. Each box shows the current integer label of an
$m$-qubit register; the root is the source $\mathsf{Q}$ with label $q$.
All arithmetic is modulo $2^m$.
(a) The targets initially contain $b_1,\ldots,b_6$.
(b) Working from the leaves upward, subtract each target parent's value
from its child. The two edges out of $\mathsf{Q}$ are skipped; dashed edges mark
this exception, and upward arrows indicate the traversal order.
(c) Working from the root downward, add each parent's current value
to its child. Every target then contains $b_v+q$, while $\mathsf{Q}$ is unchanged.}
\label{fig:tree-addition}
\end{figure*}

Arrange the \(m\)-qubit registers \(\mathsf{Q},\mathsf{B}_1,\ldots,\mathsf{B}_r\) on a
complete binary tree \(\mathcal T\), with \(\mathsf{Q}\) at the root.
Write \(h\) for its height and \(p(v)\) for the parent of a
target vertex \(v\). The initial register labels are \(q\) at
the root and \(b_v\) at target \(\mathsf{B}_v\). All arithmetic below is
modulo \(N=2^m\).

Simply adding each parent to its children from the root downward
does not give the desired map: a deeper target receives not only
\(q\), but also the initial values of the targets along its path.
We remove these unwanted contributions by preprocessing the targets.
Starting at the leaves and working upward, subtract each target
parent from its child, skipping the edges out of \(\mathsf{Q}\).
The first target level retains its initial values, while every
deeper target becomes \(b_v-b_{p(v)}\). The upward order ensures
that each subtraction uses the parent's initial value, before
that parent is modified.

Now work from the root downward, adding each parent's current
value to its child. The first target level becomes \(b_v+q\).
At every deeper vertex, the parent has already acquired its final
value \(b_{p(v)}+q\), so the update is
\begin{equation}
  (b_v-b_{p(v)})+(b_{p(v)}+q)=b_v+q.
  \label{eq:main-tree-cancellation}
\end{equation}
Thus every target receives exactly \(q\), and \(\mathsf{Q}\) is unchanged.
Figure~\ref{fig:tree-addition} illustrates these two traversals.

To implement both traversals efficiently, it suffices to construct
the downward addition circuit. On any rooted tree, this circuit
replaces each word by the sum of the initial values along the path
from the root to that word. We call this the \emph{path-sum
permutation} and denote it by \(\Pi_{\mathcal T}\) for the full
tree. Removing \(\mathsf{Q}\) leaves a forest \(\mathcal T'\) of at most
two target trees. Let \(\Pi_{\mathcal T'}\) apply the path-sum
permutation to these trees separately, leaving \(\mathsf{Q}\) unchanged.
Reversing this circuit gives precisely the upward subtraction
used for preprocessing. Hence
\begin{equation}
  \MADD_{m,r}=\Pi_{\mathcal T}\Pi_{\mathcal T'}^{-1}.
  \label{eq:main-madd}
\end{equation}
A fast path-sum circuit therefore supplies both parts of the
simultaneous adder, with the same depth and size bounds for the
inverse circuit.

\paragraph{Fourier addition one bit at a time.}
Using an \(O(m)\)-depth Fourier adder separately at each tree
level would give depth \(O(mh)\). To overlap the levels, we use
the fact that output bit \(j\) of a sum depends only on input
bits \(0,\ldots,j\). A parent's low bits can therefore enter
its children's additions while its higher bits are still being
decoded.

Consider a parent value \(a\) and an initial child value \(b\),
with \(s=a+b\pmod N\). Denote the binary digits of \(s\) by \(s_i\).
Encode the child as \(\widetilde F_m|b\rangle\). By
Eq.~\eqref{eq:fourier-qubit}, its qubit \(j\) has relative phase
\(e^{2\pi i b/2^{j+1}}\).
For each \(0\le i\le j\), the addition uses a controlled phase
of angle \(2\pi 2^i/2^{j+1}\) from parent bit \(i\) to child
qubit \(j\). Their combined contribution is
\(e^{2\pi i a/2^{j+1}}\), giving
\[
  e^{2\pi i b/2^{j+1}}e^{2\pi i a/2^{j+1}}
  =e^{2\pi i s/2^{j+1}}.
\]
Parent bits above \(j\) are unnecessary because their phase
contributions are integer multiples of \(2\pi\).

Suppose the lower sum bits \(s_0,\ldots,s_{j-1}\) have already
been decoded into computational encoding.
For each \(i<j\), use sum bit \(i\) to control a phase of angle
\(-2\pi 2^i/2^{j+1}\) on qubit \(j\).
These internal corrections remove the lower sum bits' contribution,
leaving relative phase
\begin{equation}
  \exp\!\left[
    \frac{2\pi i}{2^{j+1}}
    \left(s-\sum_{i<j}2^is_i\right)
  \right]=(-1)^{s_j}.
  \label{eq:main-bit-decoding}
\end{equation}
A Hadamard then maps
\((|0\rangle+(-1)^{s_j}|1\rangle)/\sqrt2\) to \(|s_j\rangle\).
Starting with \(j=0\), which needs no lower-bit correction,
this decodes the sum from low bits to high bits.

The addition phases and internal corrections can be interleaved
as their control bits become available; all required phases on
qubit \(j\) must precede its Hadamard.
Once decoded, bit \(j\) can correct bits \(j+1,j+2,\ldots\)
of its own word and supply addition phases to bits
\(j,j+1,\ldots\) of each child.
Thus the same decoded bit advances both the decoding within its
word and the addition in the next tree level.

\paragraph{Overlapping the tree levels.}

\begin{figure*}[t]
\centering
\resizebox{\textwidth}{!}{%
\begin{tikzpicture}[x=1cm,y=1cm,
  registerlabel/.style={
    font=\fontsize{10}{12}\selectfont,anchor=south},
  decodeedge/.style={-{Stealth[length=2mm]},
    draw=inkblue,line width=.6pt}]
% Four horizontal bands per register; bit 0 is at the bottom.
% At the end of round t=5, the decoded counts at depths 1,2,3
% are 3,2,2, since bit j at depth d is decoded at t=d+2j.
% The root is in computational encoding from the start.
\foreach \xx/\yy/\vv/\ready in {
  7.7/0/0/4,
  3.3/-1.50/1/3,12.1/-1.50/2/3,
  1.1/-3.00/3/2,5.5/-3.00/4/2,
  9.9/-3.00/5/2,14.3/-3.00/6/2,
  0/-4.50/7/2,2.2/-4.50/8/2,
  4.4/-4.50/9/2,6.6/-4.50/10/2,
  8.8/-4.50/11/2,11.0/-4.50/12/2,
  13.2/-4.50/13/2,15.4/-4.50/14/2}{
  \coordinate (reg\vv) at (\xx,\yy);
  \coordinate (top\vv) at (\xx,\yy+.36);
  \coordinate (bottom\vv) at (\xx,\yy-.36);

  \fill[lightamber]
    (\xx-.79,\yy-.36) rectangle (\xx+.79,\yy+.36);
  \fill[inkblue!35!white]
    (\xx-.79,\yy-.36)
    rectangle (\xx+.79,\yy-.36+.18*\ready);
  \draw[draw=inkblue,line width=.6pt]
    (\xx-.79,\yy-.36) rectangle (\xx+.79,\yy+.36);

  \foreach \bit in {1,2,3}{
    \ifnum\bit<\ready
      \draw[inkblue!55,line width=.35pt]
        (\xx-.79,\yy-.36+.18*\bit)--
        (\xx+.79,\yy-.36+.18*\bit);
    \else
      \draw[amber!55,line width=.35pt]
        (\xx-.79,\yy-.36+.18*\bit)--
        (\xx+.79,\yy-.36+.18*\bit);
    \fi
  }

  \ifnum\vv=0
    \node[registerlabel] at (\xx,\yy+.40) {$\mathsf{Q}$};
  \else
    \ifodd\vv
      \node[registerlabel,anchor=south east]
        at (\xx-.15,\yy+.40) {$\mathsf{B}_{\vv}$};
    \else
      \node[registerlabel,anchor=south west]
        at (\xx+.15,\yy+.40) {$\mathsf{B}_{\vv}$};
    \fi
  \fi
}
\foreach \parent/\child in {
  0/1,0/2,1/3,1/4,2/5,2/6,
  3/7,3/8,4/9,4/10,5/11,5/12,6/13,6/14}{
  \draw[decodeedge] (bottom\parent)--(top\child);
}

\node[font=\fontsize{10}{12}\selectfont,anchor=west]
  at (-.79,.88) {After logical round $t=5$ ($m=4$)};

\draw[draw=inkblue,line width=.55pt,fill=inkblue!35!white]
  (.45,-5.55) rectangle (.83,-5.30);
\node[font=\fontsize{10}{12}\selectfont,anchor=west]
  at (1.02,-5.425) {Bits available as controls};

\draw[draw=amber,line width=.55pt,fill=lightamber]
  (9.15,-5.55) rectangle (9.53,-5.30);
\node[font=\fontsize{10}{12}\selectfont,anchor=west]
  at (9.72,-5.425) {Bits still in Fourier form};
\end{tikzpicture}%
}
\caption{The tree pipeline after logical round $t=5$ for
four-qubit words. Each box has one band per bit, with bit $0$ at
the bottom. Blue bands are in computational encoding and can
control phase gates; pale orange bands remain in Fourier form.
The root $\mathsf{Q}$ is available from the start. The first, second, and
third target levels have decoded three, two, and two bits,
respectively. Adjacent nonroot levels alternate their Hadamards:
the first and third target levels have just decoded a new bit in
round $5$, whereas the second target level last decoded a bit
in round $4$.
The exact schedule is proved in Appendix~\ref{sec:simultaneous-addition}.}
\label{fig:pipeline}
\end{figure*}
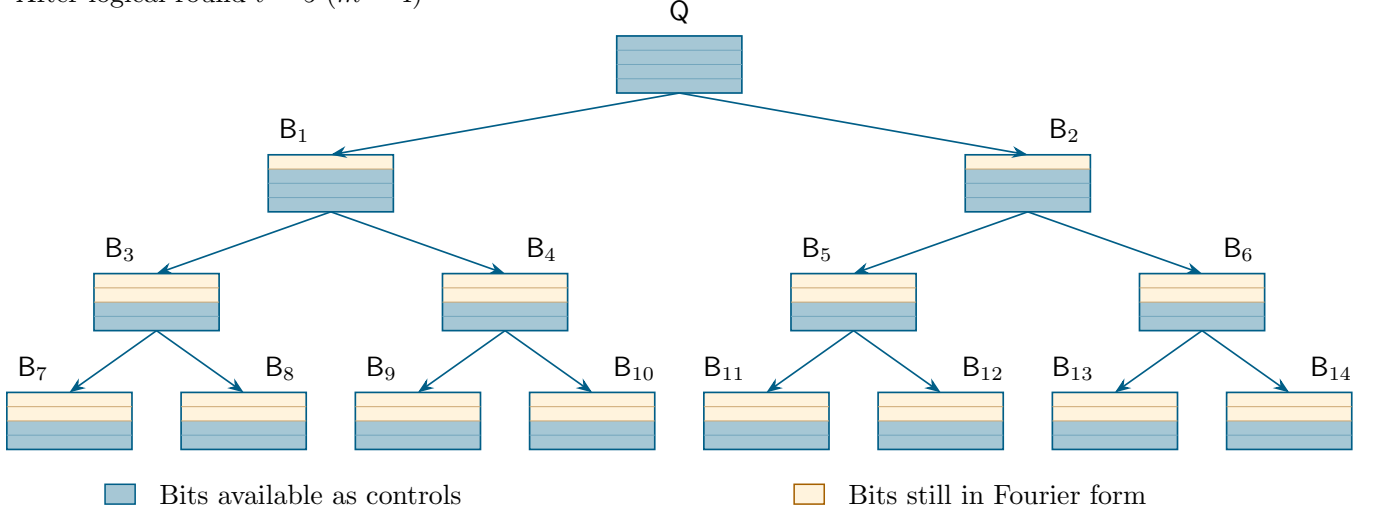

For the full-tree path-sum circuit, first apply \(\widetilde F_m\)
to all target words in parallel, leaving \(\mathsf{Q}\) in computational
encoding. This initialization has depth \(O(m)\).
Number subsequent logical rounds from \(t=0\).
Each round has four substeps: internal corrections, addition
phases to left children, addition phases to right children,
and finally the scheduled Hadamards.
The assignment in Appendix~\ref{sec:simultaneous-addition}
makes the gates within each substep disjoint, so every round
has constant depth.

The schedule alternates Hadamards between neighboring nonroot
levels: each child decodes bit \(j\) one round after its parent,
and successive bits of the same word are decoded two rounds apart.
To see how the rounds connect, suppose bit \(j<m-1\) of a word
has just been decoded.
Its next bit has already received all contributions from control
positions below \(j\).
In the next round, it corrects that word's bit \(j+1\) and adds
to bit \(j\) of each child. The word also receives the parent's
bit-\(j\) contribution to its bit \(j+1\).
In that same round, a nonroot parent decodes its bit \(j+1\).
One round later, this parent bit supplies the final addition
phase needed for the word's bit \(j+1\), which then undergoes
its own Hadamard. If the parent is \(\mathsf{Q}\), its bits are already
available and the same schedule applies.
Previously decoded bits continue supplying phases, moving their
target positions one step higher each round as this cycle repeats.
Figure~\ref{fig:pipeline} shows the resulting overlap at the
end of round \(5\) for four-bit words.

A word at depth \(d\ge1\) starts by decoding bit \(0\) in
round \(d\), then decodes one further bit every two rounds.
The last Hadamard is therefore applied by round \(h+2m-2\).
Together with the initialization, this gives depth \(O(m+h)\).
Fourier encoding, addition phases, and decoding require
\(O(m^2)\) gates per target word, so the full-tree circuit has
size \(O(rm^2)\). The same construction applies in parallel to
the target trees in \(\mathcal T'\), with each tree's root left
in computational encoding. Reversing that circuit implements
the preprocessing \(\Pi_{\mathcal T'}^{-1}\) with the same
resource bounds. The two circuits in Eq.~\eqref{eq:main-madd}
therefore give total depth \(O(m+h)=O(m+\log r)\) and size
\(O(rm^2)\), using only \(\mathsf{Q},\mathsf{B}_1,\ldots,\mathsf{B}_r\), as claimed
in Lemma~\ref{thm:simultaneous-addition}.

\section{Symmetric Boolean Functions}
\label{sec:symmetric-functions}

A Boolean function \(f:\bits^n\to\bits\) is symmetric if its
value depends only on the Hamming weight of its input. Thus there
is a function \(g:\{0,\ldots,n\}\to\bits\) such that
\begin{equation}
  f(x)=g(|x|).
  \label{eq:symmetric-table}
\end{equation}
We evaluate \(f\) reversibly by preserving the input and adding
\(f(x)\) modulo two to a target qubit. Let \(\mathsf{X}\) be the
\(n\)-qubit input register, \(\mathsf{T}\) the target, and \(\mathsf{W}\) an
initially zero weight register of
\(m=\lceil\log_2(n+1)\rceil\) qubits.

\begin{theorem}[Exact symmetric Boolean evaluation]
  \label{thm:symmetric-oracle}
  For \(n\ge2\) and every symmetric Boolean function
  \(f:\bits^n\to\bits\), there is an exact circuit
  \begin{equation}
    \begin{aligned}
      &|x\rangle_\mathsf{X}|t\rangle_\mathsf{T}|0^m\rangle_\mathsf{W}\longmapsto
      |x\rangle_\mathsf{X}|t\oplus f(x)\rangle_\mathsf{T}|0^m\rangle_\mathsf{W}
    \end{aligned}
    \label{eq:symmetric-oracle-map}
  \end{equation}
  of depth \(O(\log n)\), size \(O(n\log n)\), and width
  \(n+m+1\). The only initially clean workspace is \(\mathsf{W}\).
  All other temporary qubits are borrowed from \(\mathsf{X}\) and
  restored exactly.
\end{theorem}

The construction has three stages: compute the Hamming weight
into \(\mathsf{W}\), evaluate \(g\) on that weight to flip \(\mathsf{T}\),
and reverse the weight computation. The counting circuit is
provided by Theorem~\ref{thm:hamming-weight}. The remaining task
is to evaluate \(g\) using workspace borrowed from \(\mathsf{X}\),
and to restore \(\mathsf{X}\) before reversing the count. We use the
following general Boolean evaluation circuit.

\begin{lemma}[Boolean evaluation with dirty workspace]
  \label{thm:boolean-oracle}
  For every \(k\ge1\) and every Boolean function
  \(h:\bits^k\to\bits\), there is an exact circuit implementing
  \begin{equation}
    |y\rangle_\mathsf{Y}|t\rangle_\mathsf{T}
    \longmapsto
    |y\rangle_\mathsf{Y}|t\oplus h(y)\rangle_\mathsf{T},
    \label{eq:boolean-oracle-operator}
  \end{equation}
  where \(\mathsf{Y}\) is a \(k\)-qubit input register and \(\mathsf{T}\) is
  a disjoint target qubit. The circuit has depth \(O(k)\) and
  size \(O(2^k)\) and requires no initially clean ancillary
  qubits. Its additional workspace satisfies
  \begin{equation}
    \text{number of dirty qubits}\le 2^k-k-1.
    \label{eq:boolean-dirty-count}
  \end{equation}
  All borrowed qubits are restored exactly, including any initial
  entanglement with other registers.
\end{lemma}

We first use this lemma to prove
Theorem~\ref{thm:symmetric-oracle}, including the bound on
borrowed space. We then describe the evaluation circuit and its
routing pipeline. The complete schedule and proof of
Lemma~\ref{thm:boolean-oracle} appear in
Appendix~\ref{subsec:boolean-oracle}.

\subsection{Symmetric functions from Hamming weight}
\label{subsec:symmetric-reduction}

Applying Lemma~\ref{thm:boolean-oracle} to all \(m\) weight
bits would allow up to \(2^m-m-1\) dirty qubits, which may
exceed the \(n\) available input qubits. To fit the evaluation
inside \(\mathsf{X}\), separate the most significant weight bit. Put
\(k=m-1\), and write
\[
  \mathsf{C}:=\mathsf{W}_k,\qquad \mathsf{Y}:=(\mathsf{W}_0,\ldots,\mathsf{W}_{k-1}).
\]
For basis values \(c\in\bits\) and \(0\le y<2^k\),
the value in \(\mathsf{W}\) is \(2^kc+y\). Here \(y\) is the
integer represented by the lower \(k\) bits.

Extend \(g\) by setting \(g(z)=0\) for \(n<z<2^m\);
these values are never reached by a valid Hamming weight. Define
two Boolean functions of the lower bits:
\begin{align}
  g_0(y)&:=g(y),\notag\\
  g_\Delta(y)&:=g(y)\oplus g(2^k+y).
  \label{eq:two-branch-functions}
\end{align}
Then
\begin{equation}
  g(2^kc+y)=g_0(y)\oplus c\,g_\Delta(y).
  \label{eq:two-branch-identity}
\end{equation}
The first term can be added directly to \(\mathsf{T}\). To add the
second term, borrow one qubit \(\mathsf{A}\) from \(\mathsf{X}\), and choose
the workspace for Lemma~\ref{thm:boolean-oracle} from the
remaining input qubits. The total borrowed space is at most
\begin{equation}
  (2^k-k-1)+1\le 2^k=2^{m-1}\le n.
  \label{eq:symmetric-capacity}
\end{equation}
Thus all required workspace fits inside \(\mathsf{X}\), and the same
qubits can be reused for each application of the lemma.

Apply the following operations in the listed order:
\begin{enumerate}[leftmargin=*,itemsep=2pt,topsep=4pt]
  \item Evaluate \(g_0\) on \(\mathsf{Y}\), adding its value to \(\mathsf{T}\).
  \item Evaluate \(g_\Delta\) on \(\mathsf{Y}\), adding its value to \(\mathsf{A}\).
  \item Apply \(\operatorname{CCX}_{\mathsf{C},\mathsf{A}\to \mathsf{T}}\).
  \item Repeat step 2, restoring \(\mathsf{A}\).
  \item Repeat step 3.
\end{enumerate}
Each function evaluation uses Lemma~\ref{thm:boolean-oracle}
and restores its workspace. If \(a\) is the initial basis
value of \(\mathsf{A}\), the two Toffoli gates add
\begin{equation}
  c\bigl(a\oplus g_\Delta(y)\bigr)\oplus ca
  =c\,g_\Delta(y)
  \label{eq:symmetric-accumulator-cancellation}
\end{equation}
to \(\mathsf{T}\). The unknown initial value of the borrowed qubit
therefore cancels. Together with the first evaluation, this gives
\begin{equation}
  \begin{aligned}
    &|x\rangle_\mathsf{X}|t\rangle_\mathsf{T}|c\rangle_\mathsf{C}|y\rangle_\mathsf{Y}\\
    &\qquad\longmapsto
    |x\rangle_\mathsf{X}|t\oplus g(2^kc+y)\rangle_\mathsf{T}
    |c\rangle_\mathsf{C}|y\rangle_\mathsf{Y}.
  \end{aligned}
  \label{eq:symmetric-evaluation-operator}
\end{equation}
In particular, the entire input register \(\mathsf{X}\) is restored
before the weight computation is reversed.

\begin{proof}[Proof of Theorem~\ref{thm:symmetric-oracle}]
  Compute the Hamming weight into \(\mathsf{W}\), perform the above
  evaluation, and reverse the counting circuit. Equation~
  \eqref{eq:symmetric-evaluation-operator} preserves \(\mathsf{X},\mathsf{W}\)
  and adds \(g(|x|)=f(x)\) to \(\mathsf{T}\). The inverse count
  therefore returns \(\mathsf{W}\) to \(|0^m\rangle\).

  The three function evaluations have total depth \(O(k)\)
  and size \(O(2^k)\), and the two Toffoli gates have
  constant cost. Since \(k=O(\log n)\) and \(2^k\le n\),
  this stage has depth \(O(\log n)\) and size \(O(n)\).
  Combining it with the forward and inverse counts gives depth
  \(O(\log n)\) and size \(O(n\log n)\). By
  Eq.~\eqref{eq:symmetric-capacity}, all borrowed qubits lie
  in \(\mathsf{X}\). The width is therefore \(n+m+1\), with only
  \(\mathsf{W}\) required as clean workspace.
\end{proof}

For comparison, Ref.~\cite{ZiNieSun2025} achieves depth
\(O(\log^2 n)\) with the same \(m\) clean workspace qubits,
whereas the \(O(\log n)\)-depth construction of
Ref.~\cite{ZiYuanNieZhang2026} uses \(O(n/\log n)\) clean
ancillary qubits. Theorem~\ref{thm:symmetric-oracle} retains
logarithmic depth using only the weight register as clean
workspace. Since parity is symmetric and requires depth
\(\Omega(\log n)\) in this gate model~\cite[Theorem~3.2]{FangEtAl2006},
the worst-case depth over symmetric Boolean functions is optimal.

\subsection{Boolean evaluation by pipelined routing}
\label{subsec:dirty-lookup-mechanism}

\begin{figure*}[t]
\centering
\resizebox{\textwidth}{!}{%
\begin{tikzpicture}[x=1cm,y=1cm,
  font=\fontsize{9}{11}\selectfont,
  channel/.style={draw=inkblue,line width=.55pt,fill=white,
    minimum width=1.08cm,minimum height=.52cm,inner sep=3pt},
  router/.style={draw=inkblue,line width=.55pt,fill=lightblue,
    minimum width=1.08cm,minimum height=.52cm,inner sep=3pt},
  gate/.style={draw=inkblue,line width=.6pt,fill=lightblue,
    minimum width=1.08cm,minimum height=.58cm,inner sep=3pt},
  leafgate/.style={gate,minimum width=.43cm,
    minimum height=.43cm,inner sep=2pt},
  flipgate/.style={draw=amber,fill=lightamber},
  identitygate/.style={draw=black!35,fill=white},
  routeevent/.style={gate,minimum width=.72cm,
    minimum height=.36cm,inner sep=1pt,
    font=\fontsize{9}{11}\selectfont},
  depositevent/.style={routeevent,draw=amber,fill=lightamber},
  routehighlight/.style={draw=inkblue,line width=.95pt,
    fill=inkblue!25!white},
  deposithighlight/.style={draw=amber,line width=.95pt,
    fill=amber!23!white},
  selected/.style={-{Stealth[length=2mm]},
    draw=inkblue,line width=1.05pt},
  ordinary/.style={-{Stealth[length=1.7mm]},
    draw=black!40,line width=.55pt},
  heading/.style={anchor=west,
    font=\fontsize{10}{12}\selectfont\bfseries},
  note/.style={align=center,inner sep=0pt}]

\node[heading] at (0,5.95)
  {(a) Routing by the input: $(y_0,y_1)=(1,0)$};
\node[heading] at (8.85,5.95)
  {(b) Evaluation and restoration};
\draw[black!22,line width=.5pt] (8.37,.26)--(8.37,5.72);

% The root has only a router; B is an external borrowed wire.
\node[channel,dashed,minimum width=1.9cm] (borrowed) at (3.00,5.05)
  {$\mathsf{B}$ (borrowed)};
\node[router,minimum width=2.2cm] (root) at (5.62,5.05)
  {$\mathsf{R}_{\emptyset}:y_0=1$};
\draw[densely dashed,inkblue!65] (root.west)--(borrowed.east);
\node[note] at (5.62,4.56)
  {Root: routing qubit only};
\node[channel] (m0) at (1.35,3.57) {$\mathsf{M}_0$};
\node[router] (r0) at (2.74,3.57) {$\mathsf{R}_0$};
\node[channel] (m1) at (4.91,3.57) {$\mathsf{M}_1$};
\node[router,minimum width=1.62cm] (r1) at (6.53,3.57)
  {$\mathsf{R}_1:y_1=0$};
\draw[black!25,line width=.5pt,rounded corners=2pt]
  (.66,3.13) rectangle (3.44,4.00);
\draw[black!25,line width=.5pt,rounded corners=2pt]
  (4.22,3.13) rectangle (7.47,4.00);
\draw[densely dashed,inkblue!65] (r0.west)--(m0.east);
\draw[densely dashed,inkblue!65] (r1.west)--(m1.east);
\draw[ordinary] (borrowed.south)--node[pos=.55,left] {$0$}(m0.north);
\draw[selected] (borrowed.south)--node[pos=.55,right] {$1$}(m1.north);
\node[channel] (m00) at (.77,2.12) {$\mathsf{M}_{00}$};
\node[channel] (m01) at (2.67,2.12) {$\mathsf{M}_{01}$};
\node[channel,fill=lightblue,line width=1.05pt]
  (m10) at (4.77,2.12) {$\mathsf{M}_{10}$};
\node[channel] (m11) at (6.87,2.12) {$\mathsf{M}_{11}$};
\draw[ordinary] (m0.south)--node[pos=.46,left] {$0$}(m00.north);
\draw[ordinary] (m0.south)--node[pos=.46,right] {$1$}(m01.north);
\draw[selected] (m1.south)--node[pos=.46,left] {$0$}(m10.north);
\draw[ordinary] (m1.south)--node[pos=.46,right] {$1$}(m11.north);
\foreach \leaf/\xx/\mask/\gatestyle in {
  00/.77/X/flipgate,01/2.67/I/identitygate,
  10/4.77/X/flipgate,11/6.87/I/identitygate}{
  \draw[black!45,line width=.5pt] (m\leaf.south)--(\xx,1.51);
  \node[leafgate,\gatestyle] at (\xx,1.30) {$\mask$};
}
\node[note] at (4.00,.64)
  {Leaf gates $\mathcal L_h$ encode the fixed function $h$.\\
   The blue path carries the state of $\mathsf{B}$ to $\mathsf{M}_{10}$.};

% The four gate boxes are executed from left to right.
\node[note] at (13.13,5.16)
  {Route, apply leaf gates, reverse route};
\node[note] at (13.13,4.63)
  {$\mathcal Q_h=\mathcal R_k^{-1}\mathcal L_h\mathcal R_k$,
   \qquad $\mathcal Q_h^2=I$};
\node[gate] (q1) at (9.61,3.78) {$\mathcal Q_h$};
\node[gate,flipgate,minimum width=1.65cm] (k1) at (11.74,3.78)
  {$\operatorname{CNOT}_{\mathsf{B}\to \mathsf{T}}$};
\node[gate] (q2) at (13.87,3.78) {$\mathcal Q_h$};
\node[gate,flipgate,minimum width=1.65cm] (k2) at (16.00,3.78)
  {$\operatorname{CNOT}_{\mathsf{B}\to \mathsf{T}}$};
\draw[ordinary,draw=inkblue] (q1.east)--(k1.west);
\draw[ordinary,draw=inkblue] (k1.east)--(q2.west);
\draw[ordinary,draw=inkblue] (q2.east)--(k2.west);
\node[note] at (13.13,4.26)
  {Execution order $\longrightarrow$};
\node[note,inner sep=3pt] (changed) at (10.54,2.69)
  {$\mathsf{B}:b\oplus h(y)$\\tree may change};
\node[note,inner sep=3pt] (restored) at (15.02,2.69)
  {$\mathsf{B}:b$\\tree restored};
\draw[inkblue!55,line width=.5pt]
  (q1.south)--(9.61,3.15)--(changed.north);
\draw[inkblue!55,line width=.5pt]
  (q2.south)--(13.87,3.15)--(restored.north);
\node[note] at (13.13,1.77)
  {Control: $\mathsf{B}$; target: $\mathsf{T}$};
\node[note,draw=inkblue!65,fill=lightblue,line width=.6pt,
  minimum width=8.08cm,minimum height=.78cm,inner sep=4pt]
  at (13.13,.70)
  {Net target flip: $h(y)$.\\
   $\mathsf{Y},\mathsf{B}$ and the entire dirty tree are restored.};

\draw[black!22,line width=.5pt] (0,-.02)--(17.22,-.02);
\node[heading] at (0,-.47) {(c) Routing pipeline: $k=5$};
\node[anchor=east] at (17.22,-.47)
  {New states enter every $2$ rounds};

% A route crosses level r in round 2j+r; Y_j is deposited in round 3j.
% The B row has state index j=5 and no deposit.
\foreach \t in {11,12}{
  \fill[black!6] (1.90+.99*\t-.495,-4.635)
    rectangle (1.90+.99*\t+.495,-1.32);
}
\node[anchor=east] at (1.20,-1.04) {Round $t$};
\foreach \t in {0,...,14}{
  \node at (1.90+.99*\t,-1.04) {$\t$};
  \draw[black!16,line width=.35pt]
    (1.90+.99*\t-.495,-1.455)--(1.90+.99*\t-.495,-4.635);
}
\draw[black!16,line width=.35pt] (16.255,-1.455)--(16.255,-4.635);
\foreach \j in {0,...,5}{
  \ifnum\j<5
    \node[anchor=east] at (1.20,-1.72-.53*\j) {$\mathsf{Y}_{\j}$};
  \else
    \node[anchor=east] at (1.20,-1.72-.53*\j) {$\mathsf{B}$};
  \fi
  \draw[black!16,line width=.35pt]
    (1.405,-1.455-.53*\j)--(16.255,-1.455-.53*\j);
}
\draw[black!16,line width=.35pt] (1.405,-4.635)--(16.255,-4.635);
\foreach \j/\r/\t in {
  1/0/2,
  2/0/4,2/1/5,
  3/0/6,3/1/7,3/2/8,
  4/0/8,4/1/9,4/2/10,4/3/11,
  5/0/10,5/1/11,5/2/12,5/3/13,5/4/14}{
  \ifnum\t<11
    \node[routeevent]
      at (1.90+.99*\t,-1.72-.53*\j) {$r{=}\r$};
  \else
    \ifnum\t<13
      \node[routeevent,routehighlight]
        at (1.90+.99*\t,-1.72-.53*\j) {$r{=}\r$};
    \else
      \node[routeevent]
        at (1.90+.99*\t,-1.72-.53*\j) {$r{=}\r$};
    \fi
  \fi
}
\foreach \j/\t in {0/0,1/3,2/6,3/9,4/12}{
  \ifnum\j=4
    \node[depositevent,deposithighlight]
      at (1.90+.99*\t,-1.72-.53*\j) {$j{=}\j$};
  \else
    \node[depositevent]
      at (1.90+.99*\t,-1.72-.53*\j) {$j{=}\j$};
  \fi
}
\node[anchor=west] at (0,-5.03)
  {Blue: route $r\to r+1$;\qquad orange: deposit at depth $j$.};
\node[anchor=east] at (17.22,-5.03)
  {Shaded columns: disjoint operations.};
\end{tikzpicture}%
}
\caption{Boolean evaluation by routing with dirty workspace.
(a) The input $(y_0,y_1)=(1,0)$ is stored in routers on the
selected path. The state of $\mathsf{B}$ follows the blue path to leaf
$10$; bits are listed in routing order, so the integer input is
$y=1$. Each nonroot internal vertex has a channel $\mathsf{M}_u$ and a
router $\mathsf{R}_u$. The root has only a router, and leaves have only
channels. Dashed links indicate router control. The leaf gates
encode $h(0,0)=h(1,0)=1$ and $h(0,1)=h(1,1)=0$.
All borrowed qubits may start dirty.
(b) Two applications of $\mathcal Q_h$, each followed by
$\operatorname{CNOT}_{\mathsf{B}\to \mathsf{T}}$, add $h(y)$ to the target and
restore $\mathsf{B}$ and the tree. Operations execute from left to right.
(c) Each row tracks the state initially on $\mathsf{Y}_j$ or $\mathsf{B}$, rather
than a fixed physical wire. A blue box marked $r$ routes the
state from depth $r$ to $r+1$; an orange box marked $j$ deposits
it in a router at depth $j$. New states start at two-round
intervals, and each moving state advances one level per round.
The shaded columns show operations with disjoint supports:
in round $11$, $\mathsf{Y}_4$ crosses level $3$ while $\mathsf{B}$ crosses level
$1$; in round $12$, $\mathsf{Y}_4$ is deposited at depth $4$ while $\mathsf{B}$
crosses level $2$.}
\label{fig:lookup}
\end{figure*}

We now explain the construction behind
Lemma~\ref{thm:boolean-oracle}. A routing tree sends a borrowed
state to a leaf selected by the input, where a fixed gate encodes
the function value. Two uses of this operation write the answer
to the target and restore the borrowed space. Overlapping the
address loading and routing then gives depth \(O(k)\).

\paragraph{Routing to the selected leaf.}
Fix \(h:\bits^k\to\bits\), with input register
\(\mathsf{Y}=(\mathsf{Y}_0,\ldots,\mathsf{Y}_{k-1})\) and target qubit \(\mathsf{T}\).
For a basis input, let \(y_j\) be the bit on \(\mathsf{Y}_j\) and
identify the bit string with the integer
\(y=\sum_{j=0}^{k-1}2^jy_j\). Thus the bits are routed
from least to most significant, and
\(h(y)=h(y_0,\ldots,y_{k-1})\).

Take a full binary tree of depth \(k\). A vertex at depth
\(j\) has a label \(u\in\bits^j\) recording the path from
the root, with \(0\) for left and \(1\) for right. Every
nonroot vertex has a channel qubit \(\mathsf{M}_u\), and every internal
vertex has a routing qubit \(\mathsf{R}_u\), called its router. The root
has only \(\mathsf{R}_\emptyset\), and leaves have only channels.
Borrow all these tree qubits and an additional qubit \(\mathsf{B}\),
disjoint from \(\mathsf{Y},\mathsf{T}\) and the tree. Their initial states are
arbitrary and may be entangled with other registers.

At an internal vertex, two controlled-SWAPs route an incoming
state: first exchange it with the left child's channel if the
router is \(0\), then with the right child's channel if the
router is \(1\). The incoming wire is \(\mathsf{M}_u\) at a nonroot
vertex and an external source wire at the root. A routing layer
at level \(r\) applies these pairs at every vertex of depth
\(r\), including those outside the selected path. It moves
the selected state from depth \(r\) to depth \(r+1\).

To understand the route, first load the address sequentially.
Swap \(\mathsf{Y}_0\) into \(\mathsf{R}_\emptyset\). For each
\(1\le j<k\), route the state of \(\mathsf{Y}_j\) through levels
\(0,\ldots,j-1\), then swap \(\mathsf{M}_u\) with \(\mathsf{R}_u\) at
every vertex of depth \(j\). We call this last step a
\emph{deposit}. On the selected path, it stores \(y_j\) in
the router; the other deposits exchange dirty states. The
separate channel and router let later states pass through a
vertex while its address bit remains in place. Finally, route
the state of \(\mathsf{B}\) through all \(k\) levels. It reaches the
leaf with path label \(y_0\cdots y_{k-1}\), independently
of the initial dirty values. In Fig.~\ref{fig:lookup}(a),
\((y_0,y_1)=(1,0)\) selects leaf \(10\), corresponding
to the integer input \(y=1\).

\paragraph{Evaluation and restoration.}
Let \(\Route_k\) denote the forward circuit just described.
At each leaf \(u\), apply \(X\) if \(h(u)=1\), then
reverse the route. Define
\begin{equation}
  \begin{aligned}
    \mathcal L_h&:=
      \prod_{\substack{u\in\bits^k\\h(u)=1}}X_{\mathsf{M}_u},\\
    \Query_h&:=\Route_k^{-1}\mathcal L_h\Route_k.
  \end{aligned}
  \label{eq:main-raw-query}
\end{equation}
The function \(h\) is a fixed classical specification: its
values determine which leaf gates appear in the circuit.

On a computational-basis input, \(\Query_h\) returns \(\mathsf{Y}\)
unchanged and changes \(\mathsf{B}\) from \(b\) to
\(b\oplus h(y)\). The routing bits remain in place until
their reverse deposits, so reversing the route brings the
modified state back along the same path. Other dirty tree
qubits may change because the leaf gates also act outside that
path. Since \(\mathcal L_h^2=I\), however,
\(\Query_h^2=I\): two applications restore \(\mathsf{B}\) and the
entire tree.

To write the function value into \(\mathsf{T}\), execute the following
four operations from left to right:
\begin{equation}
  \Query_h
  \;\longrightarrow\;
  \operatorname{CNOT}_{\mathsf{B}\to \mathsf{T}}
  \;\longrightarrow\;
  \Query_h
  \;\longrightarrow\;
  \operatorname{CNOT}_{\mathsf{B}\to \mathsf{T}}.
  \label{eq:main-lookup-cleanup}
\end{equation}
The CNOTs add the two successive values of \(\mathsf{B}\) to \(\mathsf{T}\),
giving the net flip
\[
  \bigl(b\oplus h(y)\bigr)\oplus b=h(y).
\]
The unknown value \(b\) cancels, as in
Eq.~\eqref{eq:symmetric-accumulator-cancellation}. Since the
CNOTs change only \(\mathsf{T}\), the two applications of \(\Query_h\)
restore \(\mathsf{B}\) and the tree exactly. These gates are exact
basis permutations, so the same action and restoration hold
for arbitrary superpositions and entangled inputs.
Figure~\ref{fig:lookup}(b) shows this sequence.

\paragraph{Pipelining the routing.}
Loading each input state completely before starting the next
would take depth \(O(k^2)\). Instead, start a new state every
two logical rounds. Each moving state advances by one routing
level per round. After reaching its designated depth, an
address state is deposited in the following round; \(\mathsf{B}\)
enters last and continues to a leaf.

Index the states initially on \(\mathsf{Y}_0,\ldots,\mathsf{Y}_{k-1}\) by
\(j=0,\ldots,k-1\), and assign index \(j=k\) to the
state initially on \(\mathsf{B}\). A routing event has
\(0\le r<j\le k\), and a deposit has \(0\le j<k\).
Use the schedule
\begin{equation}
  \begin{aligned}
    \text{Route state }j\text{ across level }r &: \quad t=2j+r,\\
    \text{Deposit }\mathsf{Y}_j\text{ at depth }j &: \quad t=3j.
  \end{aligned}
  \label{eq:main-routing-time}
\end{equation}
Thus \(\mathsf{Y}_0\) is deposited directly at the root in round \(0\).
For \(j\ge1\), state \(j\) begins routing in round
\(2j\); an address state completes its last routing step in
round \(3j-1\) and is deposited one round later.

The two-round spacing keeps simultaneous routing layers at
least two levels apart. A layer at level \(r\) uses routers
at depth \(r\) and channels at depths \(r\) and \(r+1\),
with the external source replacing the root channel. Such
layers therefore have disjoint supports. Deposits fit between
successive uses of the channels: at a nonroot depth \(j\),
state \(j\) arrives in round \(3j-1\), is deposited in
round \(3j\), and the next state arrives in round \(3j+1\).
That next state uses the depth-\(j\) router only in round
\(3j+2\), after the address bit has been stored. All operations
assigned to a round consequently have disjoint supports, and
their order on each shared qubit agrees with sequential
loading. Appendix~\ref{subsec:boolean-oracle} gives the
formal verification.

Figure~\ref{fig:lookup}(c) shows the schedule for \(k=5\).
In round \(11\), \(\mathsf{Y}_4\) crosses level \(3\) while
\(\mathsf{B}\) crosses level \(1\). In round \(12\), \(\mathsf{Y}_4\)
is deposited at depth \(4\) while \(\mathsf{B}\) crosses level
\(2\). Thus \(\mathsf{B}\) can enter in round \(10\), before
the final address deposit: it needs the depth-\(4\) router
only in round \(14\).

The last routing step of \(\mathsf{B}\) occurs in round \(3k-1\),
so the forward circuit uses \(3k\) logical rounds. Each
routing pair and deposit has constant exact depth in our gate
model, and disjoint operations run in parallel. Hence
\(\Route_k\), its inverse, and the full evaluation circuit
all have depth \(O(k)\). Loading \(\mathsf{Y}_j\) uses
\(O(2^j)\) gates, while routing \(\mathsf{B}\) and applying the
leaf gates each use \(O(2^k)\) gates. Summing over the
input states gives total size \(O(2^k)\), including all
operations outside the selected path.

This full-tree construction uses \(3\cdot2^k-2\) dirty
qubits, including \(\mathsf{B}\). Appendix~\ref{subsec:boolean-oracle}
reduces this to at most \(2^k-k-1\), while preserving depth
\(O(k)\) and size \(O(2^k)\), to obtain the bound in
Lemma~\ref{thm:boolean-oracle}.

\section{Dicke-state preparation}
\label{sec:dicke-states}

For a fixed classical weight $k\in\{0,\ldots,n\}$, the Dicke state is
\begin{equation}
  |D_k^n\rangle=\binom nk^{-1/2}\sum_{|x|=k}|x\rangle.
  \label{eq:dicke-definition}
\end{equation}
The Hamming-weight counter gives two ways to act on these states:
measuring its output projects onto a fixed-weight subspace, while
phasing and uncomputing its output implements a reflection about
that subspace. We apply these operations to existing adaptive and
unitary preparation strategies. Both constructions use only
$m=\lceil\log_2(n+1)\rceil$ additional qubits, for a total width
of $n+m$. Table~\ref{tab:dicke-comparison} summarizes the main
depth--space comparisons.

\subsection{Measurement-assisted preparation}
\label{subsec:dicke-adaptive}

Here we allow parallel single-qubit measurements and resets at
constant depth, and exclude classical feedforward latency from
the depth. Yu et al.~\cite{YuEtAl2026} prepare Dicke states by
alternating collective one-qubit rotations with Hamming-weight
measurements. For even $n$, their balanced-state protocol,
including its restart rule, reaches weight $n/2$ in $O(\log n)$
rounds in expectation. Replacing its weight-measurement circuit
by our counter gives the following result.

\begin{corollary}[Measurement-assisted Dicke-state preparation]
\label{cor:dicke-adaptive}
For even $n\ge2$, there is a protocol that prepares
$|D_{n/2}^n\rangle$ exactly, terminates almost surely, and has
\begin{equation}
  \mathbb E[D]=O(\log^2 n),\qquad
  \mathbb E[S]=O(n\log^2 n).
  \label{eq:dicke-adaptive-expected}
\end{equation}
It uses $m$ reusable ancillary qubits, reset to zero after each
weight measurement.
\end{corollary}

\begin{proof}
Let $\Pi_w=\sum_{|x|=w}|x\rangle\langle x|$.
Theorem~\ref{thm:hamming-weight} gives
\begin{equation}
  \CHW|\psi\rangle_\mathsf{X}|0^m\rangle_\mathsf{W}
  =\sum_{w=0}^n\bigl(\Pi_w|\psi\rangle_\mathsf{X}\bigr)|w\rangle_\mathsf{W}.
  \label{eq:dicke-weight-measurement}
\end{equation}
Measuring $\mathsf{W}$ in the computational basis therefore realizes
the outcome-$w$ instrument $\rho\mapsto\Pi_w\rho\Pi_w$,
preserving coherence within each weight subspace. The measured
register is in the known state $|w\rangle$ and is reset by
parallel conditional $X$ gates, without uncomputing the counter.
Each measurement thus has depth $O(\log n)$ and size
$O(n\log n)$.

Use this instrument in the protocol of Ref.~\cite{YuEtAl2026},
retaining its rotation choices and its restart rule when
$|w-n/2|>\sqrt{n/2}$. The measurement statistics and conditional
states are unchanged, so its $O(\log n)$ expected-round bound
applies. Rotations and data resets add only constant depth and
$O(n)$ gates per round, proving the resource bounds. The collective
rotations and weight measurements preserve the permutation-symmetric
subspace, and the data resets return to it. Within that subspace,
$\Pi_{n/2}$ has the one-dimensional range spanned by
$|D_{n/2}^n\rangle$. Acceptance is therefore exact.
\end{proof}

The low-ancilla elementary-gate implementation described in
Ref.~\cite{YuEtAl2026} uses the $O(\log^2 n)$-depth counter of
Zi et al.~\cite{ZiNieSun2025}. Combined with the expected-round
bound, it gives expected depth $O(\log^3 n)$. Our substitution
reduces this to $O(\log^2 n)$ with the same $m$ ancillary qubits.
This comparison concerns that elementary-gate implementation;
the collective weight measurement available in the physical
model of Ref.~\cite{YuEtAl2026} is a different primitive.

More workspace permits other tradeoffs.
Rethinasamy et al.~\cite{RethinasamyLaBordeWilde2024} implement
the same weight measurement in depth $O(\log n)$ with $O(n)$
ancillary qubits, which also gives expected depth $O(\log^2 n)$
when inserted into Yu et al.'s protocol.

Yang et al.~\cite[Theorem~3]{YangWeiYin2026SpinSqueezing}
prepare the balanced Dicke state with $n+O(\log n)$ ancillary
qubits, depth $O(\log n)$ per attempt, and success probability
bounded below by a positive constant.
Their final exact weight measurement projects the
permutation-symmetric data state onto $|D_{n/2}^n\rangle$,
so successful preparation is exact.
Repetition therefore gives expected depth $O(\log n)$.
This gives a smaller expected-depth bound using linear
ancillary space, whereas our construction uses only $m$
ancillary qubits.

The heralded symmetrization algorithm of
Liu et al.~\cite[Result~1]{LiuChildsGottesman2026} prepares any
fixed-weight Dicke state with $O(n\log n)$ ancillary qubits and
expected depth $O(\log n\log\log n)$ under repetition.
Constant-depth probabilistic adaptive constructions with
polynomial workspace are also known~\cite{LuoChenMa2026Dicke};
their single-attempt depth must be considered together with
their success probability.

\begin{table*}[t]
\centering
\caption{Selected depth--space bounds for exact Dicke-state
preparation with all-to-all connectivity and elementary
one- and two-qubit gates. Here
$m=\lceil\log_2(n+1)\rceil$ and $r=\min\{k,n-k\}$.
Ancillary qubits are in addition to the $n$ output qubits.
Adaptive rows give expected total depth, including repetition
when required, counting measurements and resets but excluding
classical feedforward latency; unitary rows give circuit depth.
All entries are upper bounds. Balanced-state rows assume even $n$.
The unitary rows concern a fixed classical $k$ with
$1\le k\le n-1$; $k=0,n$ are product states.
The measurement-assisted and unitary constructions using our
counter are given in Corollaries~\ref{cor:dicke-adaptive}
and~\ref{cor:dicke-unitary}, respectively.}
\label{tab:dicke-comparison}
\begingroup
\small
\setlength{\tabcolsep}{7pt}
\renewcommand{\arraystretch}{1.2}
\begin{tabular*}{\textwidth}{@{\extracolsep{\fill}}lccc@{}}
\toprule
Construction & Target weight & Ancillary qubits & Depth \\
\midrule
\multicolumn{4}{@{}l}{\textit{Measurement-assisted preparation}} \\
Yu et al.\ + prior counter~\cite{YuEtAl2026,ZiNieSun2025}
  & $k=n/2$ & $m$ & $O(\log^3 n)$ \\
Yu et al.\ + our counter
  & $k=n/2$ & $m$ & $O(\log^2 n)$ \\
Yang et al.~\cite{YangWeiYin2026SpinSqueezing}
  & $k=n/2$ & $n+O(\log n)$ & $O(\log n)$ \\
\midrule
\multicolumn{4}{@{}l}{\textit{Unitary preparation}} \\
Yuan and Zhang~\cite{YuanZhang2025Dicke}
  & Any fixed $k$ & $0$
  & $O\!\left(r+\log(r+1)\log(n/r)\right)$ \\
Amplitude amplification + our counter
  & Any fixed $k$ & $m$ & $O(r^{1/4}\log n)$ \\
Liu et al.~\cite[Result~2]{LiuChildsGottesman2026}
  & Any fixed $k$ & $O(n\log n)$
  & $O(\log^3 n\log\log n)$ \\
\bottomrule
\end{tabular*}
\endgroup
\end{table*}

\subsection{Unitary preparation}
\label{subsec:dicke-unitary}

We now return to the measurement-free circuit model.
Building on amplitude-amplification approaches to Dicke-state
preparation~\cite{PiroliStyliarisCirac2024,NagibSaffmanMolmer2025},
we combine collective one-qubit rotations with exact fixed-weight
reflections. Our counter implements these reflections using
elementary gates and logarithmic ancillary space.

\begin{corollary}[Unitary Dicke-state preparation]
\label{cor:dicke-unitary}
For any fixed classical $k$ with $1\le k\le n-1$, let
$r=\min\{k,n-k\}$. There is an exact unitary circuit preparing
$|D_k^n\rangle$ with
\begin{equation}
  D=O(r^{1/4}\log n),\qquad
  S=O(nr^{1/4}\log n),
  \label{eq:dicke-unitary-resources}
\end{equation}
using $m$ ancillary qubits, returned to zero.
\end{corollary}

\begin{proof}
Write $R_s=I_\mathsf{X}-2\Pi_s$. Computing the weight, phasing the
output, and uncomputing gives
\begin{equation}
\begin{split}
  &\CHW^\dagger
  \bigl[I_\mathsf{X}\otimes(I_\mathsf{W}-2|s\rangle\langle s|_\mathsf{W})\bigr]\CHW
  \bigl(|\psi\rangle_\mathsf{X}|0^m\rangle_\mathsf{W}\bigr)\\
  &\hspace{3em}=(R_s|\psi\rangle_\mathsf{X})|0^m\rangle_\mathsf{W}.
\end{split}
\label{eq:dicke-weight-reflection}
\end{equation}
The phase on $\mathsf{W}$ has depth and size $O(m)$ using dirty qubits
borrowed from $\mathsf{X}$, all restored exactly~\cite[Lemma~7.2]{BarencoEtAl1995}.
Hence each $R_s$ has depth $O(\log n)$ and size $O(n\log n)$,
with no ancillary qubits beyond $\mathsf{W}$.

To see explicitly why the preparation can be exact, let
$A_p$ be a layer of identical one-qubit rotations preparing
\begin{equation}
  |\psi_p\rangle=A_p|0^n\rangle
  =\bigl(\sqrt{1-p}|0\rangle+\sqrt p|1\rangle\bigr)^{\otimes n}.
\end{equation}
Its squared overlap with $|D_k^n\rangle$ is
$P_k(p)=\binom nk p^k(1-p)^{n-k}$, whose maximum is
$P_k(k/n)=\Theta(r^{-1/2})$.
Set
\begin{equation}
  \theta_* = \arcsin\sqrt{P_k(k/n)},\qquad
  q=\left\lceil\frac{\pi}{4\theta_*}-\frac12\right\rceil.
\end{equation}
By continuity, one can choose $p\in[0,k/n]$ such that
$P_k(p)=\sin^2[\pi/(4q+2)]$.
Within the permutation-symmetric subspace, $R_k$ is the
reflection that changes the sign of $|D_k^n\rangle$, while
$R_0=I_\mathsf{X}-2|0^n\rangle\langle0^n|$. Standard amplitude
amplification therefore gives
\begin{equation}
  \bigl(-A_pR_0A_p^\dagger R_k\bigr)^q|\psi_p\rangle
  =|D_k^n\rangle.
\end{equation}
The overall sign of each iterate is an irrelevant global phase.
Since $q=O(r^{1/4})$, each iteration uses two reflections and
two parallel rotation layers, giving the stated bounds.
\end{proof}

Yuan and Zhang~\cite{YuanZhang2025Dicke} give an ancilla-free
unitary construction with depth
$O(r+\log(r+1)\log(n/r))$.
Our depth bound is asymptotically smaller when
$r=\omega((\log n)^{4/3})$; for balanced states, the bounds are
$O(n)$ and $O(n^{1/4}\log n)$, respectively.
The total width increases only from $n$ to
$n+\lceil\log_2(n+1)\rceil=(1+o(1))n$.
For smaller weights one may instead use their construction;
selecting between the two gives depth
\begin{equation}
  O\!\left(\min\left\{
    r+\log(r+1)\log(n/r),\ r^{1/4}\log n
  \right\}\right)
  \label{eq:dicke-combined-depth}
\end{equation}
with at most $m$ ancillary qubits.
This comparison is for fixed classical $k$; it does not assert
a coherent preparation map for a superposition of weights.

With substantially more workspace, lower depth is possible.
The exact unitary symmetrization algorithm of
Liu et al.~\cite[Result~2]{LiuChildsGottesman2026} has depth
$O(\log^3 n\log\log n)$ and uses $O(n\log n)$ ancillary qubits;
it also handles coherent superpositions of weights.
Constant-depth unitary constructions with polynomial workspace
and unbounded-arity gates are available as
well~\cite{JoshiVasconcelos2026Dicke,GrettaGuptaJoshi2026Dicke}.
Decomposing these unbounded-arity gates into
elementary gates using additional polynomial workspace gives
$O(\log n)$ depth. Thus the unitary application above provides
a depth improvement over an ancilla-free construction while
keeping total width $n+O(\log n)$.

\section{Discussion}
\label{sec:discussion}

Our exact Hamming-weight circuit achieves
optimal depth $O(\log n)$ and minimum width $n+m$, with size
$O(n\log n)$. Here $m=\lceil\log_2(n+1)\rceil$ is the number
of binary output qubits. These bounds hold for arbitrary
one-qubit gates and CNOTs with all-to-all connectivity. The counter also gives exact circuits for every specified
symmetric Boolean function, with the same depth and size,
using $m$ reusable clean workspace qubits.

The counting and Boolean-evaluation constructions share a common
design: organizing the computation on a tree and pipelining
operations across its levels. This structure enables low-depth
computation with limited workspace.

The Dicke-state applications illustrate the benefits of this
approach with logarithmic ancillary space. For even $n$, the
balanced-state protocol's low-ancilla implementation improves
from $O(\log^3 n)$ to $O(\log^2 n)$ expected elementary-gate depth.
The unitary construction improves on the ancilla-free depth bound
in the weight range specified in Sec.~\ref{subsec:dicke-unitary},
with total width $n+O(\log n)$. These are improvements under
limited workspace; wider circuits can achieve lower depth.

A natural next goal is to achieve size $O(n)$ while retaining
depth $O(\log n)$ and width $n+m$ for counting. Further directions
include finding applications that turn these low-space primitives
into improvements in overall algorithmic resources, and extending
tree-based pipelines to other reversible computations.
Restricted connectivity and fault-tolerant gate sets also require
accounting for routing and rotation-synthesis costs at the desired
accuracy.

\begin{acknowledgments}
Wei Zi acknowledges support from the National Natural Science
Foundation of China (Grant No.~62602563) and the Guangdong
Provincial Quantum Science Strategic Initiative
(Grant No.~GDZX2503001).

The authors acknowledge the use of OpenAI's ChatGPT for
author-directed exploratory discussions, literature searches,
checks of derivations, and assistance with manuscript preparation,
including editing and LaTeX figure preparation.
The authors independently verified all mathematical claims and
proofs, reviewed and revised the final manuscript, and take full
responsibility for the correctness and integrity of this work.
\end{acknowledgments}

\section*{Author contributions}
W.Z. developed the main results and wrote the initial draft.
P.Y. and J.N. contributed to the development of the circuit
constructions through discussions and to the writing and
revision of the manuscript.

\section*{Competing interests}
The authors declare no competing interests.

\section*{Data availability}
This is a theoretical study. The constructions and proofs
supporting its findings are provided in the main text and
appendices.

\bibliography{ref}

\appendix
\section{Exact schedules for simultaneous modular addition}
\label{sec:simultaneous-addition}

This appendix gives the exact pipeline underlying
Sec.~\ref{subsec:tree-mechanism} and proves
Lemma~\ref{thm:simultaneous-addition}.
The schedule has two features: neighboring tree levels alternate
their decoding steps, and each decoded bit continues to supply
phase gates to higher bits of its own word and to its children.
Every logical round consists of a fixed number of parallel
substeps. We first establish progressive Fourier decoding, then
combine decoding and addition on a binary tree, and finally obtain
the simultaneous adder by composing two tree circuits.
Throughout, \(m\ge1\), word arithmetic is modulo \(N=2^m\), and
\(\widetilde F_m\) is the product Fourier transform in
Eq.~\eqref{eq:product-fourier}.

\subsection{Progressive Fourier decoding}

In the product Fourier representation, lower bits can be decoded
before higher bits. Once decoded, a bit controls the phase
corrections needed by the remaining higher bits.
For \(0\le i\le j<m\), define
\begin{equation}
  \alpha_{i,j}:=\frac{2\pi 2^i}{2^{j+1}}.
  \label{eq:alpha}
\end{equation}
We use a two-round spacing between successive decoded bits;
this will allow parent and child words to alternate their
decoding steps in the tree construction.

\begin{lemma}[Progressive exact decoder]
  \label{lem:progressive-decoder}
  The inverse transform \(\widetilde F_m^{-1}\) has an exact circuit
  of depth \(O(m)\) and size \(O(m^2)\), with no ancillary qubits,
  in which output bit \(j\) is decoded at logical round \(2j\).
\end{lemma}

\begin{proof}
  After bit \(i\) is decoded at round \(2i\), it corrects bit
  \(i+1\) in the next round, bit \(i+2\) in the following round,
  and so on. Thus its correction to bit \(j>i\) occurs at
  round \(2i+(j-i)=i+j\).
  Explicitly, at round \(t=0,\ldots,2m-2\), apply all gates
  \begin{equation}
    \CP_{i,j}(-\alpha_{i,j}),
    \qquad 0\le i<j<m,\quad i+j=t,
    \label{eq:decoder-corrections}
  \end{equation}
  followed by \(H_j\) if \(t=2j\).
  At fixed \(t\), the correction pairs form a matching, since
  qubit \(k\) can only be paired with qubit \(t-k\).
  Each round therefore has constant depth.

  Fix an input \(\widetilde F_m|z\rangle\), where
  \(z=\sum_{j=0}^{m-1}2^jz_j\).
  For \(i<j\), the correction at round \(i+j\) occurs after
  \(H_i\) and before \(H_j\), since \(2i<i+j<2j\).
  Inductively, immediately before \(H_j\), qubit \(j\) is in
  the state
  \begin{equation}
    \frac{|0\rangle+(-1)^{z_j}|1\rangle}{\sqrt2},
    \label{eq:decoder-before-h}
  \end{equation}
  because its relative phase is
  \[
    \exp\!\left[
      \frac{2\pi i}{2^{j+1}}
      \left(z-\sum_{i<j}2^iz_i\right)
    \right]=(-1)^{z_j}.
  \]
  The Hadamard produces \(|z_j\rangle\) with amplitude \(+1\).
  Hence the circuit maps every \(\widetilde F_m|z\rangle\) to
  \(|z\rangle\) without a residual phase, and linearity gives
  \(\widetilde F_m^{-1}\) on arbitrary inputs.

  The circuit uses \(m(m-1)/2\) controlled phases and \(m\)
  Hadamards over \(2m-1\) rounds. Constant-cost decompositions
  of the controlled phases give the stated bounds.
  Reversing the circuit and negating its phase angles implements
  \(\widetilde F_m\) with the same resources.
\end{proof}

\subsection{The pipelined path-sum permutation}

Let \(\mathcal S\) be a rooted binary tree of \(m\)-qubit word
registers, with root \(\rho\), \(L\) nonroot vertices, and height
\(h\). For a nonroot vertex \(v\), let \(p(v)\) be its parent.
The depth \(d(v)\) is the number of edges from \(\rho\) to \(v\),
so \(d(\rho)=0\) and \(d(v)=d(p(v))+1\).
Label the at most two children of each vertex as left and right,
arbitrarily but consistently.
For initial word labels \(z_v\), define the path sums by
\begin{equation}
  s_\rho=z_\rho,
  \qquad s_v=z_v+s_{p(v)}\pmod N.
  \label{eq:path-sums}
\end{equation}
Write \(s_v=\sum_{j=0}^{m-1}2^js_{v,j}\), with
\(s_{v,j}\in\bits\), and denote qubit \(j\) at vertex \(v\)
by \((v,j)\).
The map in Eq.~\eqref{eq:path-sums} is a permutation, with inverse
\begin{equation}
  z_\rho=s_\rho,
  \qquad z_v=s_v-s_{p(v)}\pmod N.
  \label{eq:path-sum-inverse}
\end{equation}
Let \(\Pi_{\mathcal S}\) denote the corresponding basis permutation.

\begin{lemma}[Pipelined path sums]
  \label{lem:pipelined-path-sums}
  The permutation \(\Pi_{\mathcal S}\) has an exact circuit of
  depth \(O(m+h)\) and size \(O(Lm^2)\), with no ancillary qubits.
\end{lemma}

\begin{proof}
  If \(L=0\), the permutation is the identity and requires no
  gates. Assume henceforth that \(L\ge1\).

  \emph{Initialization and logical rounds.}
  Apply \(\widetilde F_m\) to all nonroot words in parallel,
  leaving the root in computational encoding.
  Lemma~\ref{lem:progressive-decoder} gives an initialization
  depth of \(O(m)\).
  Number subsequent logical rounds from \(t=0\), after initialization.
  Each round consists of the following four substeps, in order:
  \begin{enumerate}[label=(\roman*),nosep]
    \item internal decoding corrections within each word;
    \item addition phases from each parent to its left child;
    \item addition phases from each parent to its right child;
    \item the Hadamards scheduled for this round.
  \end{enumerate}
  Within each substep, all scheduled gates are applied in parallel.
  The precise assignment below makes their supports disjoint.

  \emph{Alternating decoding and the local cycle.}
  We schedule \(H_{v,j}\), the Hadamard on \((v,j)\), at round
  \(d(v)+2j\). Thus successive bits of the same word are decoded
  two rounds apart. A nonroot parent decodes each bit one round
  before its child, so neighboring nonroot words alternate
  their Hadamards. The root needs no Hadamards: all its bits are
  available as addition controls from the start.

  Once a nonroot bit \((v,j)\) is decoded at round
  \(T=d(v)+2j\), it continues to control phase gates.
  In round \(T+\ell\), for \(\ell\ge1\), it corrects bit
  \(j+\ell\) of its own word and adds to bit \(j+\ell-1\)
  of each child. Thus both target positions advance by one per
  round; the first child target is the same bit position \(j\).
  Table~\ref{tab:path-sum-rounds} records this local rule.
  Gates with nonexistent targets are omitted.

  \begin{table}[tb]
    \centering
    \caption{Targets controlled by bit \((v,j)\) after its
    Hadamard in round \(T=d(v)+2j\). Internal corrections have
    negative angles, and additions to children have positive
    angles. The last row gives the general rule for \(\ell\ge1\);
    only target indices in \(0,\ldots,m-1\) are used.}
    \label{tab:path-sum-rounds}
    \begin{tabular}{@{}l@{\qquad}l@{\qquad}l@{}}
      \toprule
      Round & Internal target & Target in each child \\
      \midrule
      \(T+1\) & \(j+1\) & \(j\) \\
      \(T+2\) & \(j+2\) & \(j+1\) \\
      \(T+\ell\) & \(j+\ell\) & \(j+\ell-1\) \\
      \bottomrule
    \end{tabular}
  \end{table}

  To see how the next bit takes over, suppose \(j<m-1\).
  By the end of round \(T\), all contributions to \((v,j+1)\)
  from control positions below \(j\) have arrived.
  In round \(T+1\), this qubit receives the internal correction
  from \((v,j)\) and the addition phase from \((p(v),j)\).
  If the parent is nonroot, its bit \(j+1\) is decoded at the
  end of that same round, since
  \(d(p(v))+2(j+1)=T+1\).
  In round \(T+2\), that parent bit supplies the final addition
  phase to \((v,j+1)\), which then undergoes its own Hadamard.
  If the parent is the root, the required bit was already available.
  Meanwhile, each child decodes its bit \(j\) in round \(T+1\).
  The next bit can therefore enter the same cycle while earlier
  bits continue supplying phases to higher positions.

  \emph{Exact gate assignment.}
  For each nonroot \(v\), define addition gates for
  \(0\le i\le j<m\) and internal corrections for \(0\le i<j<m\)
  by
  \begin{align}
    A_{v;i,j}
      &:=\CP_{(p(v),i),(v,j)}(+\alpha_{i,j}),\notag\\
    C_{v;i,j}
      &:=\CP_{(v,i),(v,j)}(-\alpha_{i,j}).
    \label{eq:path-sum-gates}
  \end{align}
  The local rule above gives the schedule
  \begin{equation}
    \begin{aligned}
      A_{v;i,j}:&\quad t=d(v)+i+j,\\
      C_{v;i,j}:&\quad t=d(v)+i+j,\\
      H_{v,j}:&\quad t=d(v)+2j.
    \end{aligned}
    \label{eq:path-sum-schedule}
  \end{equation}
  Addition gates controlled by the root follow this same schedule.
  In particular, each nonroot vertex at depth \(d\) decodes its lowest
  bit in round \(d\), so the pipeline starts at the root and
  advances one tree level per round.

  \emph{Availability and disjoint supports.}
  For \(i<j\), the control of \(C_{v;i,j}\) is decoded
  \(j-i\ge1\) rounds before the correction is applied.
  If \(p(v)\ne\rho\), the control of \(A_{v;i,j}\) is decoded
  before its use by
  \begin{equation}
    (d(v)+i+j)-(d(v)-1+2i)=1+j-i\ge1.
    \label{eq:parent-bit-available}
  \end{equation}
  Root controls are available throughout.
  All phase gates targeting \((v,j)\) occur by round
  \(d(v)+2j\). Only \(A_{v;j,j}\) occurs in that final round,
  and it precedes \(H_{v,j}\) within the round.

  At fixed \(t\), internal corrections in word \(v\) satisfy
  \(i+j=t-d(v)\), so they form a matching.
  The addition gates on any fixed parent--child edge likewise
  form a matching between the two words.
  To check adjacent edges, observe that a nonroot qubit \((v,j)\)
  receives addition phases only in rounds
  \(t\le d(v)+2j\), whereas it controls additions to its children
  only in rounds \(t\ge d(v)+2j+1\).
  It therefore cannot receive and send an addition phase in the
  same round. Within either child substep, each vertex has at
  most one incoming and one outgoing edge; the preceding
  observation excludes collisions between them.
  Thus each addition substep has disjoint qubit supports, even
  when its edges share a word register.
  Finally, the Hadamards act on distinct qubits.
  Each of the four substeps has constant depth after elementary-gate
  decomposition, and hence so does each logical round.

  \emph{Exactness.}
  Fix a computational-basis input and induct over the Hadamard
  events. Before \(H_{v,j}\), the addition controls have values
  \(s_{p(v),i}\), and the internal correction controls have values
  \(s_{v,i}\). The relative phase of qubit \((v,j)\) is
  \begin{equation}
    \exp\!\left[
      \frac{2\pi i}{2^{j+1}}
      \left(
        z_v+\sum_{i=0}^{j}2^is_{p(v),i}
        -\sum_{i=0}^{j-1}2^is_{v,i}
      \right)
    \right].
    \label{eq:path-sum-decoding-phase}
  \end{equation}
  Equation~\eqref{eq:path-sums} implies that the expression in
  parentheses is congruent to \(2^js_{v,j}\) modulo \(2^{j+1}\).
  The phase is therefore \((-1)^{s_{v,j}}\), and the Hadamard
  produces \(|s_{v,j}\rangle\) with amplitude \(+1\).
  The root is unchanged. Every computational-basis input is thus
  mapped to its path sums with no residual phase, so linearity
  gives the exact unitary \(\Pi_{\mathcal S}\) on arbitrary inputs,
  including inputs entangled with an external reference.

  \emph{Resources.}
  The last scheduled round is at most \(h+2m-2\).
  Including the parallel Fourier initialization, the total depth
  is therefore \(O(m+h)\).
  Each nonroot word uses \(m(m+1)/2\) addition gates,
  \(m(m-1)/2\) internal corrections, and \(m\) Hadamards,
  together with \(O(m^2)\) gates for its initialization.
  Summing over the \(L\) nonroot words gives size \(O(Lm^2)\).
  No qubits beyond the tree's word registers are used.
\end{proof}

\subsection{Completing the simultaneous adder}

\begin{proof}[Proof of Lemma~\ref{thm:simultaneous-addition}]
  Arrange \(\mathsf{Q},\mathsf{B}_1,\ldots,\mathsf{B}_r\) on the complete binary tree
  \(\mathcal T\) of Sec.~\ref{subsec:tree-mechanism}, with
  \(\mathsf{Q}\) at the root. Removing \(\mathsf{Q}\) leaves a forest
  \(\mathcal T'\) with at most two components.
  Let \(\Pi_{\mathcal T'}\) apply the path-sum permutation to
  each component separately, leaving \(\mathsf{Q}\) unchanged.

  For initial labels \(q,b_1,\ldots,b_r\),
  Eq.~\eqref{eq:path-sum-inverse} shows that
  \(\Pi_{\mathcal T'}^{-1}\) leaves each first-level target
  at \(b_v\) and changes every deeper target to
  \(b_v-b_{p(v)}\).
  Applying \(\Pi_{\mathcal T}\) next gives \(b_v+q\) at the
  first target level. Inductively, every deeper target becomes
  \[
    (b_v-b_{p(v)})+(b_{p(v)}+q)=b_v+q\pmod N.
  \]
  Hence, in the usual operator order,
  \begin{equation}
    \MADD_{m,r}=\Pi_{\mathcal T}\Pi_{\mathcal T'}^{-1},
    \label{eq:appendix-madd-factorization}
  \end{equation}
  as in Eq.~\eqref{eq:main-madd}.

  Apply Lemma~\ref{lem:pipelined-path-sums} to \(\mathcal T\)
  and to each component of \(\mathcal T'\).
  The forest components run in parallel, with depths measured
  from their own roots. These roots remain in computational
  encoding, and only their descendants are Fourier-encoded.
  Reversing and inverting the forest circuit implements
  \(\Pi_{\mathcal T'}^{-1}\) with the same resource bounds.

  The full tree has \(r\) nonroot words, and the forest
  components together have at most \(r\) nonroot words.
  The two successive circuits therefore have total size
  \(O(rm^2)\) and depth \(O(m+h)\).
  Since the complete tree has height
  \(h=\lfloor\log_2(r+1)\rfloor\) and \(m,r\ge1\), the depth
  is \(O(m+\log r)\).
  Only \(\mathsf{Q},\mathsf{B}_1,\ldots,\mathsf{B}_r\) are used, completing the proof.
\end{proof}

\section{Boolean evaluation with fully dirty workspace}
\label{subsec:boolean-oracle}

We prove Lemma~\ref{thm:boolean-oracle} using the construction
introduced in Sec.~\ref{subsec:dirty-lookup-mechanism}. We first
give a circuit of depth \(O(k)\) and size \(O(2^k)\) using
\(3\cdot2^k-2\) dirty qubits. We then reduce the workspace
to at most \(2^k-k-1\) qubits without changing the asymptotic
depth or size. All borrowed qubits are restored exactly,
including their correlations with the input and any external
reference system.

\subsection{Tree registers and sequential routing}
\label{subsec:boolean-routing}

Fix \(k\ge1\) and a Boolean function \(h:\bits^k\to\bits\).
Let \(\mathsf{Y}=(\mathsf{Y}_0,\ldots,\mathsf{Y}_{k-1})\) be the input register and
\(\mathsf{T}\) the target qubit. For a computational-basis input, write
\(y_j\) for the bit on \(\mathsf{Y}_j\) and identify the bit string
with \(y=\sum_{j=0}^{k-1}2^jy_j\). Thus
\(h(y)=h(y_0,\ldots,y_{k-1})\), and the routing order is
from least to most significant bit.

Use a full binary tree of depth \(k\). A vertex at depth
\(j\) has a path label \(u\in\bits^j\), with root label
\(\emptyset\); its children are \(u0\) and \(u1\).
Each nonroot vertex has a channel qubit \(\mathsf{M}_u\), and each
internal vertex has a router \(\mathsf{R}_u\). The root therefore has
only \(\mathsf{R}_\emptyset\), while leaves have only channels.
Borrow all tree qubits, together with an external qubit \(\mathsf{B}\)
disjoint from \(\mathsf{Y},\mathsf{T}\) and the tree. There are \(2^k-1\)
routers and \(2^{k+1}-2\) channels, so the total number of
borrowed qubits is
\begin{equation}
  (2^k-1)+(2^{k+1}-2)+1=3\cdot2^k-2.
  \label{eq:routing-tree-count}
\end{equation}
No initial value is prescribed for any of these qubits.

At a nonroot internal vertex \(u\), perform the following
two gates in the displayed order, from left to right:
\begin{equation}
  \CSWAP_{\overline{\mathsf{R}_u}}(\mathsf{M}_u,\mathsf{M}_{u0})
  \;\longrightarrow\;
  \CSWAP_{\mathsf{R}_u}(\mathsf{M}_u,\mathsf{M}_{u1}).
  \label{eq:routing-pair}
\end{equation}
The first gate is controlled on \(\mathsf{R}_u=0\), and the second
on \(\mathsf{R}_u=1\). This pair swaps the incoming channel with the
child selected by the router and preserves the router itself.
At the root, replace \(\mathsf{M}_u\) by the external source wire
currently being routed, either \(\mathsf{Y}_j\) or \(\mathsf{B}\).
A routing layer at level \(r\) applies the pair at every
vertex of depth \(r\), including vertices outside the
selected path. The pairs within one layer have disjoint
supports.

First define the forward circuit by sequential address loading.
Swap \(\mathsf{Y}_0\) into \(\mathsf{R}_\emptyset\). For each
\(j=1,\ldots,k-1\), route the state initially on \(\mathsf{Y}_j\)
through levels \(0,\ldots,j-1\), then apply
\(\operatorname{SWAP}(\mathsf{M}_u,\mathsf{R}_u)\) at every vertex of depth
\(j\). These parallel SWAPs form the depth-\(j\)
\emph{deposit}. The depth-\(0\) deposit is the initial
swap into \(\mathsf{R}_\emptyset\). After all deposits, route the
state initially on \(\mathsf{B}\) through levels \(0,\ldots,k-1\).
Denote the resulting forward circuit by \(\Route_k\).

For each basis input, immediately after the depth-\(j\)
deposit the router on the path \(y_0\cdots y_{j-1}\)
contains \(y_j\). This follows by induction: the preceding
routers direct \(\mathsf{Y}_j\)'s value to that vertex, and the
deposit stores it in the router. Subsequent routing uses this
router only as a control, so its stored value remains in
place. The other routing pairs and deposits merely exchange
dirty values and do not alter this induction. Consequently
\(\Route_k\) sends the value initially on \(\mathsf{B}\) to leaf
\(y_0\cdots y_{k-1}\), independently of every initial
dirty value.

\subsection{The two-round routing pipeline}
\label{subsec:boolean-pipeline}

Index the states initially on \(\mathsf{Y}_0,\ldots,\mathsf{Y}_{k-1}\) by
\(j=0,\ldots,k-1\), and use \(j=k\) for the state
initially on \(\mathsf{B}\). A routing event for state \(j\) crosses
level \(r\), where \(0\le r<j\le k\). Schedule the
complete routing layers and deposits by
\begin{equation}
  \begin{array}{l|c}
    \text{Operation} & \text{Round} \\\hline
    \text{Route state }j\text{ across level }r & 2j+r \\
    \text{Deposit }\mathsf{Y}_j\text{ at depth }j & 3j
  \end{array}
  \label{eq:routing-time}
\end{equation}
Deposits have \(0\le j<k\); there is no deposit for \(\mathsf{B}\).
This is Eq.~\eqref{eq:main-routing-time}, illustrated in
Fig.~\ref{fig:lookup}(c). New states start at two-round
intervals, and each moving state crosses one level per round.
We verify disjointness, availability of controls, and
equivalence to the sequential circuit separately.

\paragraph{Disjoint supports within each round.}
A routing layer at level \(r\) uses routers at depth \(r\)
and channels at depths \(r\) and \(r+1\), with an external
source in place of the depth-\(0\) channel. Two distinct
routing layers in the same round have levels differing by a
nonzero even integer, and hence have disjoint supports.

A deposit at a nonroot depth \(j\) can overlap only routing
levels \(j-1\) and \(j\). A route of state \(s\) across
level \(j-1\) in the deposit round would require
\(2s+j-1=3j\), or \(2(s-j)=1\), which is impossible.
A route across level \(j\) in that round would require
\(s=j\), contrary to the routing condition \(j<s\).
The root deposit occurs in round \(0\), before any routing
event. Distinct deposits occur in distinct rounds. Thus all
complete routing pairs and deposit SWAPs assigned to a round
act on pairwise disjoint sets of qubits.

\paragraph{Availability of the stored address bits.}
For \(1\le j<k\), state \(j\) reaches depth \(j\) in
round \(3j-1\) and is deposited in round \(3j\).
Every later state \(s>j\) reaches the channels at that
depth in round \(2s+j-1\ge3j+1\), and uses the routers
there in round \(2s+j\ge3j+2\). For the next state, these
four events therefore occur in the order
\[
  \underbrace{3j-1}_{\text{arrival}}
  \;<\;
  \underbrace{3j}_{\text{deposit}}
  \;<\;
  \underbrace{3j+1}_{\text{next arrival}}
  \;<\;
  \underbrace{3j+2}_{\text{next use}}.
\]
The root router is deposited in round \(0\) and first used
in round \(2\). At every depth, all controlled uses of a
router follow its deposit, and the router is not a SWAP
target again during the forward circuit.

\paragraph{Equivalence to sequential loading.}
For each state, the routing levels occur in order, followed
by its deposit when present. Now consider a route of state
\(j\) at level \(r\) and a route of state \(s>j\) at
level \(q\). Overlapping supports require \(|q-r|\le1\), so
\begin{equation}
  (2s+q)-(2j+r)\ge2(s-j)-1\ge1.
  \label{eq:routing-order}
\end{equation}
Thus every overlapping route for the earlier state still
precedes the later one.

For a deposit at nonroot depth \(j\), no state with index
less than \(j\) has a routing event touching its qubits.
The final route of state \(j\) precedes the deposit, and
every overlapping route of a later state follows it, by the
availability inequalities above. The root deposit likewise
precedes all uses of its router. Deposits at different depths
have disjoint supports. The order of every pair of operations
sharing a qubit therefore agrees with sequential loading.
The pipeline only reorders operations with disjoint supports
and implements exactly the same circuit \(\Route_k\).

\paragraph{Depth and size.}
The final route of \(\mathsf{B}\) occurs in round \(3k-1\), so
the schedule uses \(3k\) logical rounds. Each routing pair
and deposit SWAP has constant exact cost in the one-qubit-gate
and CNOT model, without additional workspace. An open control
is obtained by conjugating its control qubit with \(X\).
The disjoint operations in each round can therefore be
implemented at constant depth, giving depth \(O(k)\) for
\(\Route_k\) and its inverse.

Counting operations at every vertex, the numbers of routing
pairs for the address states, routing pairs for \(\mathsf{B}\), and
deposit SWAPs, respectively, are
\begin{align}
  \sum_{j=1}^{k-1}\sum_{r=0}^{j-1}2^r
    &=2^k-k-1,\notag\\
  \sum_{r=0}^{k-1}2^r&=2^k-1,\notag\\
  \sum_{j=0}^{k-1}2^j&=2^k-1.
  \label{eq:routing-gate-count}
\end{align}
Hence the size is \(O(2^k)\), including operations outside
the selected path.

\subsection{Boolean evaluation and exact restoration}
\label{subsec:boolean-restoration}

Using the notation of Eq.~\eqref{eq:main-raw-query}, define
the leaf operation and the routed query by
\begin{equation}
  \begin{aligned}
    \mathcal L_h&:=
      \prod_{\substack{u\in\bits^k\\h(u)=1}}X_{\mathsf{M}_u},\\
    \Query_h&:=\Route_k^{-1}\mathcal L_h\Route_k.
  \end{aligned}
  \label{eq:raw-query}
\end{equation}
The values of the fixed function \(h\) determine which
leaf gates are present. They are not stored in an additional
quantum register. The leaf operation has at most \(2^k\)
gates and depth at most one.

\begin{lemma}[Routed query]
  \label{lem:routed-query}
  For every computational-basis assignment to \(\mathsf{Y}\), \(\mathsf{B}\),
  and the dirty tree, \(\Query_h\) preserves \(\mathsf{Y}\) and
  changes \(\mathsf{B}\) from its initial value \(b\) to
  \(b\oplus h(y)\). Other dirty tree qubits may change,
  but \(\Query_h^2=I\). The query has depth \(O(k)\)
  and size \(O(2^k)\).
\end{lemma}

\begin{proof}
  Fix a basis assignment and track the values exchanged by
  the SWAPs. The forward circuit stores the address bits in
  the routers on the selected path and sends \(\mathsf{B}\)'s value
  to leaf \(y_0\cdots y_{k-1}\). The leaf operation flips
  that value exactly when \(h(y)=1\).

  To justify the reverse route, consider every router,
  including those outside the selected path. Its controlled
  uses all follow its deposit, and it is not subsequently a
  SWAP target in the forward circuit. Its deposited value
  therefore remains unchanged until the leaf operation,
  which acts only on channels at the leaves. During reversal,
  all controlled uses of this router occur before its inverse
  deposit. They consequently see the same control value as
  in the forward circuit. This argument applies at every
  vertex, regardless of the deposited dirty value.

  All inverse routing gates thus retrace the forward swaps.
  An inverse deposit can return a flipped off-path value to
  a router only after all controlled uses of that router have
  finished. Each address value was stored in a router when
  the leaf gates acted, so it returns unchanged to its
  original input wire. The value from \(\mathsf{B}\) returns to
  \(\mathsf{B}\), with flip \(h(y)\). Values on other marked leaves
  can return to dirty tree wires, so one query need not
  restore the whole tree. Nevertheless,
  \[
    \Query_h^2=\Route_k^{-1}\mathcal L_h^2\Route_k=I.
  \]
  The depth and size bounds follow from those of
  \(\Route_k\), its inverse, and the leaf operation.
\end{proof}

Now execute the following sequence from left to right:
\begin{equation}
  \Query_h
  \;\longrightarrow\;
  \operatorname{CNOT}_{\mathsf{B}\to \mathsf{T}}
  \;\longrightarrow\;
  \Query_h
  \;\longrightarrow\;
  \operatorname{CNOT}_{\mathsf{B}\to \mathsf{T}}.
  \label{eq:basic-query-cleanup}
\end{equation}
For each basis input, the two CNOTs see \(\mathsf{B}\)'s values
\(b\oplus h(y)\) and \(b\), and hence flip \(\mathsf{T}\) by
\[
  \bigl(b\oplus h(y)\bigr)\oplus b=h(y).
\]
The CNOTs do not change \(\mathsf{Y}\), \(\mathsf{B}\), or the tree, so
the two queries restore their entire basis assignment by
\(\Query_h^2=I\). Thus the sequence implements
Eq.~\eqref{eq:boolean-oracle-operator} and restores every
borrowed qubit. All gates are exact basis permutations, with
no additional phases. By linearity, the same evaluation and
restoration hold for arbitrary joint input states, including
states entangled with an external reference. The sequence has
depth \(O(k)\) and size \(O(2^k)\), using the workspace
in Eq.~\eqref{eq:routing-tree-count}.

\subsection{Reducing the workspace}
\label{subsec:boolean-workspace}

\begin{proof}[Proof of Lemma~\ref{thm:boolean-oracle}]
  For \(k\ge2\), treat \(\mathsf{Y}_0,\mathsf{Y}_1\), the first two bits
  in routing order, separately. Put
  \begin{equation}
    \ell:=k-2,\qquad \mathsf{Z}:=(\mathsf{Y}_2,\ldots,\mathsf{Y}_{k-1}).
    \label{eq:boolean-address-split}
  \end{equation}
  For a basis input, let
  \(z=\sum_{j=2}^{k-1}2^{j-2}y_j\), so that
  \(y=y_0+2y_1+4z\). When \(\ell=0\), the register
  \(\mathsf{Z}\) is empty and \(z=0\).
  For each \(\sigma=(\sigma_0,\sigma_1)\in\bits^2\),
  define the restricted function \(h_\sigma:\bits^\ell\to\bits\)
  by
  \begin{equation}
    h_\sigma(z):=h(\sigma_0+2\sigma_1+4z).
    \label{eq:suffix-tables}
  \end{equation}

  For \(\ell\ge1\), apply the preceding routing construction
  to \(\mathsf{Z}\), using a depth-\(\ell\) tree and the same
  borrowed qubit \(\mathsf{B}\) for all four functions. The forward
  circuit \(\Route_\ell\) is independent of \(\sigma\);
  only the leaf operation changes. Write
  \[
    \Query_{h_\sigma}
      =\Route_\ell^{-1}\mathcal L_{h_\sigma}\Route_\ell.
  \]
  For \(\ell=0\), use no tree and set
  \(\Query_{h_\sigma}:=X_\mathsf{B}^{h_\sigma(0)}\).
  In both cases, the query preserves \(\mathsf{Z}\), flips \(\mathsf{B}\)
  by \(h_\sigma(z)\), and squares to the identity.
  It does not act on \(\mathsf{Y}_0,\mathsf{Y}_1\) or \(\mathsf{T}\).

  Let \(K_\sigma\) flip \(\mathsf{T}\) precisely when
  \(\mathsf{Y}_0=\sigma_0\), \(\mathsf{Y}_1=\sigma_1\), and \(\mathsf{B}=1\).
  This is an \(X\) gate with three controls, with open
  controls where \(\sigma_i=0\). It has constant exact
  cost in our gate model, preserves its controls, and needs
  no additional qubits. For each \(\sigma\), execute
  \begin{equation}
    \Query_{h_\sigma}
    \;\longrightarrow\;
    K_\sigma
    \;\longrightarrow\;
    \Query_{h_\sigma}
    \;\longrightarrow\;
    K_\sigma.
    \label{eq:two-query-cleanup}
  \end{equation}
  Only the two target flips carry the additional controls;
  both queries are executed for every input.

  Fix a basis input. If \((y_0,y_1)=\sigma\), the two
  applications of \(K_\sigma\) flip \(\mathsf{T}\) by
  \[
    \bigl(b\oplus h_\sigma(z)\bigr)\oplus b
      =h_\sigma(z).
  \]
  Otherwise both target flips are inactive. In either case,
  the queries restore \(\mathsf{Z}\), \(\mathsf{B}\), and the whole tree
  because \(\Query_{h_\sigma}^2=I\), while \(\mathsf{Y}_0,\mathsf{Y}_1\)
  remain unchanged.

  Run these blocks for
  \(\sigma=(0,0),(0,1),(1,0),(1,1)\), reusing the same
  borrowed qubits. Each block restores its workspace before
  the next starts. Exactly one block matches \((y_0,y_1)\),
  and its target flip is
  \(h_{(y_0,y_1)}(z)=h(y)\). Since every block is an exact
  basis permutation, this conclusion also holds coherently
  for arbitrary joint states and external entanglement.

  Including \(\mathsf{B}\), the workspace contains
  \begin{equation}
    3(2^\ell-1)+1=3\cdot2^{k-2}-2
    \label{eq:boolean-workspace-total}
  \end{equation}
  dirty qubits. For \(k\ge2\), the inequality
  \(2^{k-2}\ge k-1\) gives
  \begin{equation}
    \begin{aligned}
      &(2^k-k-1)-\bigl(3\cdot2^{k-2}-2\bigr)\\
      &\qquad=2^{k-2}-k+1\ge0.
    \end{aligned}
    \label{eq:boolean-workspace-bound}
  \end{equation}
  There are eight queries and eight gates \(K_\sigma\).
  Their total depth is \(O(k)\) and total size is
  \(O(2^k)\). For \(k=2\), the tree is empty, the
  workspace is just \(\mathsf{B}\), and all eight queries are
  single-qubit gates or identities, so the same bounds hold.

  Finally, for \(k=1\), apply \(X_\mathsf{T}^{h(0)}\), followed
  by a CNOT from \(\mathsf{Y}_0\) to \(\mathsf{T}\) if
  \(h(0)\oplus h(1)=1\). The identity
  \[
    h(y_0)=h(0)\oplus\bigl(h(0)\oplus h(1)\bigr)y_0
  \]
  proves correctness. This circuit uses no borrowed qubits,
  agreeing with \(2^k-k-1=0\), and has constant depth
  and size. All cases therefore satisfy the claimed bounds
  without initially clean workspace.
\end{proof}

\end{document}